\documentclass[a4paper,11pt, twoside]{article}
\usepackage{latexsym}
\usepackage[latin1]{inputenc}
\usepackage{amsmath, amsfonts, amssymb, amscd, amsthm}
\usepackage[all]{xy}
\usepackage[height=21.5cm, width=14.5cm]{geometry}
\usepackage{booktabs}
\usepackage{fancyhdr}

\newtheorem{definition}{Definition}[section]
\newtheorem{lemma}[definition]{Lemma}
\newtheorem{proposition}[definition]{Proposition}
\newtheorem{theorem}[definition]{Theorem}
\newtheorem{corollary}[definition]{Corollary}

\newcounter{app}

\theoremstyle{definition}
\newtheorem{example}[definition]{Example}
\theoremstyle{remark}
\newtheorem*{remark}{Remark}

\begin{document}

\thispagestyle{empty}
\vspace*{2cm}
\oddsidemargin10mm
\centerline{\Huge On Weighted Random Band-Matrices}
\vskip 0.3cm
\centerline{\Huge with Dependences}
\vskip 5.5cm {\Large \centerline{DISSERTATION}}
\vspace*{1cm}
{\large
\centerline{zur Erlangung des}
\centerline{Doktorgrades der
Naturwissenschaften} \centerline{an der Fakult\"at f\"ur
Mathematik und Informatik} \centerline{der Fern-Universit\"at Hagen}
\vskip 2.5cm \centerline{vorgelegt von}\vspace{0.5cm}

{\Large \centerline{Riccardo R. Catalano}}\vspace{0.5cm}}

\centerline{aus Herne}\vspace{0.5cm}
{\Large\centerline{im Februar 2016}\vspace{0.5cm}}

\newpage\thispagestyle{empty}

\vspace*{7cm} 
\begin{center}
\Large \textsc{meinen eltern gewidmet}
\end{center}

\oddsidemargin 0.7cm
\evensidemargin 0.7cm
\

\newpage\thispagestyle{empty}
\begin{abstract}
In this paper we develop techniques to compute moments of weighted random matrices $M$ the entries of which can be dependent in a certain way. This dependence is controlled via an equivalence-relation on the pairs of the indices of its entries. Every entry $\frac{1}{\sqrt{N}}\cdot X^{(N)}_{ij}$ of an $N\times N-$random matrix is multiplied with a \emph{weight} $\alpha(\frac{|i-j|}{N})$. This weight is assumed to be Riemann-integrable and to be bounded. It tourns out that the moments can be computed as a sum over integrals the kernels of which are closely connected to the weight $\alpha$. In this paper we do not only consider random band-matrices the band-width of which is proportional to its dimension but also those with a slow-growing band-width. Once being able to compute the moments of $M$ we give necessary and sufficient conditions on $\alpha$ for the Semicircle Law to be valid. Finally, we discuss weak convergence in probability of the fundamental random variable $\Delta:= \frac{1}{N}\cdot \sum_{j=1}^N \delta_{\lambda^N_j}$, where $\lambda^N_1\leq ... \leq \lambda^N_N$ denote the eigenvalues of the considered ensemble.

\vskip 1cm

In dieser Arbeit entwickeln wir Techniken zur Berechnung der Momente ge--wichteter Zufallsmatrizen $M$. Die Eintr\"age dieser Zufallsmatrizen k\"onnen Ab--h\"angigkeiten aufweisen. Die Abh\"angigkeiten werden durch eine \"Aquivalenzrelation auf den Indexpaaren der Zufallsvariablen determiniert. Jeder Eintrag $\frac{1}{\sqrt{N}}\cdot X^{(N)}_{ij}$ einer $N\times N-$Zufallsmatrix wird mit einem Riemann-integrierbaren, beschr\"ankten Gewicht $\alpha(\frac{|i-j|}{N})$ multipliziert. Es stellt sich heraus, dass die Momente von $M$ \"uber die Summe gewisser Integrale berechnet werden kann. Die Kerne der Integrale h\"angen eng mit dem Gewicht $\alpha$ zusammen. Wir betrachten in der Arbeit nicht nur zuf\"allige Bandmatrizen deren Bandbreite proportional zur Dimension w\"achst, sondern auch solche, deren Bandbreite ein schwaches Wachstum aufweist. Wir geben eine notwendige und hinreichende Bedingung an das Gewicht $\alpha$, so dass das Halbkreisgesetz (SCL) f\"ur das ensemble $M$ gilt, an. Schlie\ss lich beweisen wir schwache Konvergenz in Wahrscheinlichkeit f\"ur die fundamentale Zufallsvariable $\Delta:= \frac{1}{N}\cdot \sum_{j=1}^N \delta_{\lambda^N_j}$. Hier bezeichnen $\lambda^N_1\leq ... \leq \lambda^N_N$ die Eigenwerte von $M$.

\end{abstract}

\newpage\thispagestyle{empty}

\tableofcontents\thispagestyle{empty}
\newpage
\thispagestyle{plain}

\section*{Introduction\markboth{Introduction}{Introduction}}
\addcontentsline{toc}{section}{Introduction}

Random matrices first appeared about 90 years ago when mathematicians began to explore questions originating in statistics. About 20 years later, the most important impulse came from a physicist, \textsc{E. Wigner}. He empirically showed that resonance-spectra of heavy atoms can be approximated by eigenvalues of random matrices. Furthermore he proved the \emph{Semicircle Law} for special random matrices, see below. Throughout the recent years the \emph{eigenvalue-statistic} of random matrices was found to have a certain universality since there are various applications in mathematics and physics, see \cite{KRI}. 

In this paper we develop techniques to calculate moments of \emph{weighted} random matrices with correlations. Ensembles with correlations were considered among others by \textsc{W. Kirsch, W. Hochst\"attler} and \textsc{S. Warzel} but also by \textsc{J. Schenker} and \textsc{H. Schulz-Baldes}, see \cite{HKW} and \cite{HSBS}. We generalize their results, since Schenker's and Schulz-Baldes' results are a special case of those presented in this work. We will compute moments of the random matrix ensemble via sums of certain integrals the kernels of which are connected to the mentioned weight. First consider a probability space $(\Omega, \Sigma, \mathbb{P})$. For every $N\in\mathbb{N}$ we have real-valued random variables $X^{(N)}_{ij}\equiv X^{(N)}_{ji}$, $1\leq i,j \leq N$, with expectation zero and unit variance (finite variance is also sufficient) on $(\Omega, \Sigma, \mathbb{P})$. These random variables \emph{do not} have to be necessarily independent. The dependence of these random variables will be controlled by an equivalence relation on the pairs of their indices. Let $\sim$ be an equivalence relation on $\{ 1,...,N \}^2$. Whenever $(p,q)\not\sim (r,s)$ the random variables $X^{(N)}_{pq}$ and $X^{(N)}_{rs}$ are assumed to be independent. Of course, there will be restrictions on the equivalence relation, see conditions (\ref{AequivalenzrelationBedingung1})-(\ref{AequivalenzrelationBedingung3}), which state

\begin{eqnarray}\nonumber
\max_{p} &\#& \{ (q,r,s)\in\{ 1,...,N\}^3 | (p,q)\sim (r,s) \} = o(N^2)
\\ \nonumber \max_{p,q,r} &\#& \{ s\in\{ 1,...,N \} | (p,q)\sim(r,s) \} \leq B < +\infty
\\ \nonumber  &\#& \{ (p,q,r)\in\{ 1,...,N \}^3 | (p,q)\sim(q,r) \textnormal{ and } r\neq p \} = o(N^2) \ . 
\end{eqnarray}

\noindent An example for ensembles with dependences is

\begin{eqnarray}\nonumber
M^{(N)}:= \frac{1}{\sqrt{2N}} \cdot \left( \begin{array}{cc}
A & B \\
B^T & \pm A
\end{array} \right) \ , 
\end{eqnarray}

\noindent which will be discussed in Theorem \ref{TheoremBlockMatrizen}. Finally, we want all moments to exist and to be bounded which means

$$ \sup_N \max_{i,j} \mathbb{E}\big(|X^{(N)}_{ij}|^k\big) \leq R_k < +\infty \ \qquad \forall \ k\geq 1 \ . $$

\noindent We then consider the symmetric ensemble

$$ M^{(N)} := \frac{1}{\sqrt{N}} \left( \alpha(\frac{|i-j|}{N}) \cdot X^{(N)}_{ij} \right)_{(1\leq i,j \leq N)} \ . 
 $$

\noindent The weight-function $\alpha:[0;1]\rightarrow\mathbb{R}$ is assumed to be bounded and Riemann-integrable (throughout the whole paper). Let $\lambda^{(N)}_1\leq ... \leq \lambda^{(N)}_N$ denote the (real) eigenvalues of $M^{(N)}$. We are interested in the \emph{convergence-behaviour} of the measure-valued random variable

$$ \Delta^{(N)}:= \frac{1}{N} \cdot \sum_{i=1}^N \delta_{\lambda^{(N)}_i} \ .  $$

\noindent There are certainly different \emph{senses of convergence}. Let $(d\mu_n)_{n\in\mathbb{N}}$ be a sequence of measures on $(\mathbb{R}, \mathcal{B}(\mathbb{R}))$, where $\mathcal{B}(\mathbb{R})$ denotes the \emph{Borel-}sigma-algebra. We say that $(d\mu_n)_{n\in\mathbb{N}}$ converges to a measure $d\mu$ \emph{in the weak sense} (or \emph{in distribution}) if

$$ \int f(t) \ d\mu_n(t) \overset{n\rightarrow\infty}\longrightarrow \int f(t) \ d\mu(t) \ \qquad \ \forall \ f\in\textnormal{C}^0_b(\mathbb{R}) $$

\noindent with

$$ \textnormal{C}^0_b(\mathbb{R})  := \{ f:\mathbb{R}\rightarrow\mathbb{R} \ | f \textnormal{ is continuous and bounded} \} \ .   $$



\noindent We define

$$ \langle f, d\mu \rangle:= \int f(t) \ d\mu(t) \ .  $$

\noindent Suppose that for every $\omega\in\Omega$ we have a family $(d\mu_n^{(\omega)})_{n\in\mathbb{N}}$ of (real valued) measures and a measure $d\mu^{(\omega)}$. This family is said to converge \emph{weakly in probability} to $d\mu^{(\omega)}$ if 


$$ \mathbb{P}\big( |\langle f, d\mu_n \rangle - \langle f, d\mu \rangle| > \varepsilon  \big) \overset{n\rightarrow\infty}\longrightarrow 0 \ \forall \ \varepsilon>0 \ \textnormal{ and } f\in \textnormal{C}^0_b(\mathbb{R}) \ .  $$

\noindent Let $\sigma$ denote the \emph{semicircle density}, i.e.

$$ d\sigma(x) := \frac{1}{4\pi} \sqrt{4-x^2} \cdot\chi_{[-2;2]}(x) \ dx \ . $$

\noindent \textsc{Wigner} showed that 

\begin{eqnarray}\label{IntroHalbkreisgesetzFormulierung}
\mathbb{P}\big( |\langle f, \Delta^{(N)} \rangle - \langle f, d\sigma \rangle| > \varepsilon  \big) \overset{N\rightarrow\infty}\longrightarrow 0 \ \forall \ \varepsilon>0 \textnormal{ and } f\in\textnormal{C}^0_b   \end{eqnarray}

\noindent for the case $\alpha\equiv 1$ and independent, Bernoulli-distributed random variables $X^{(N)}_{ij}\equiv X^{(N)}_{ji}$. Because of the special form of $d\sigma$ one says that the \emph{Semicircular Law} is valid for the ensemble $M^{(N)}$. It was shown later that (\ref{IntroHalbkreisgesetzFormulierung}) also holds for any family of independent random variables with expectation zero and unit variance, see e.g. \cite{ARN}, \cite{PAS1}, \cite{TAO}. In this paper, the moments of the mentioned random variables are required to exist.

We will discuss under which conditions for $\alpha$ the SCL is valid under the above assumptions for the ensemble $M^{(N)}$, see Theorem \ref{TheoremNotwendigeUndHinreichendeBedfuerSCL}. It turns out that certain integrals play a key role for the answer. The weight function $\alpha$ turns out to be closely connected to the integral kernels of the mentioned integrals, see Theorem \ref{TheoremMomentGleichSummeIntegrale}. 

One way to prove the semicircular law (\ref{IntroHalbkreisgesetzFormulierung}) is using the \emph{moment-method}. We will also use this method to show convergence in our case. Let $\mathfrak{M}$ denote the set of all probability measures on $\mathbb{R}$ and consider the subset

$$ \mathfrak{M}^*:= \{ d\mu\in \mathfrak{M} \ | \langle |T|^k , d\mu \rangle < \infty \ \forall \ k\in\mathbb{N} \} \ . $$

\noindent $\mathfrak{M}^*$ is the set of all measures the moments of which exist and are bounded. Consider now the subset $\mathfrak{M}^{**}\subset\mathfrak{M}^{*}$ of all measures which are determined by their moments, i.e.

$$ \mathfrak{M}^{**}:= \{  d\mu\in \mathfrak{M}^* | \ \langle T^k , d\mu \rangle = \langle T^k , d\nu \rangle \ \forall \ k\in\mathbb{N} \Longrightarrow  d\mu=d\nu  \} \ .  $$

\noindent Using \textsc{Levy-Cramer}'s continuity theorem (or even \textsc{Weierstrass}' theorem) one can show that 

$$ \{ d\mu\in\mathfrak{M} \ | \textnormal{ supp}(d\mu) \textnormal{ is compact} \} \subset\mathfrak{M}^{**} \ , $$

\noindent see \cite{KRI}, \cite{KIR2}. Let us say that a measure $d\mu$ has \emph{moderately growing moments} if all moments exist and 

$$ \langle T^k , d\mu \rangle \leq c\cdot C^k\cdot k!  $$

\noindent for some constants $c, C$ and all $k\in\mathbb{N}$. For example, moments with compact support have moderately growing moments. Let $ (d\mu_n)_{n\in\mathbb{N}}$ denote a family of moderately growing moments and take $d\mu\in\mathfrak{M}$. It can be proved that

$$ \langle T^k , d\mu_n \rangle \overset{n\rightarrow\infty}\longrightarrow \langle T^k , d\mu \rangle \ \forall \ k\in\mathbb{N} \qquad \Longrightarrow \qquad d\mu_n\overset{\textnormal{weak}}\longrightarrow d\mu \ \textnormal{ and } d\mu\in\mathfrak{M}^{**} , $$

\noindent see \cite{AGZ}, \cite{KIR2}. Since supp$(d\sigma)=[-2;2]$ this measure is determined by its moments. Let $C_k$ denote the $k-$th \textsc{Catalan}-number, i.e. 

$$ C_k := \frac{1}{k+1} \cdot \left( \begin{array}{c}
2k \\ k
\end{array}\right) \ . $$

\noindent \cite{AGZ} show that for every even $k\in\mathbb{N}$ we have

$$ \langle T^k , d\sigma \rangle = C_{\frac{k}{2}} \ ,  $$

\noindent while the odd moments vanish. To prove the SCL it is then sufficient to compute the moments $\langle T^k , \Delta^{(N)} \rangle$ and discuss convergence. We initially have a look at the mentioned moments. Since $M^{(N)}$ is symmetric, there exists a transformation $S$ with

$$ S^{-1}MS = \textnormal{diag}(\lambda_1^{(N)}, ..., \lambda_N^{(N)}) =: D_\lambda \ .  $$

\noindent An explicit calculation gives

\begin{eqnarray}\nonumber
\langle T^k , \Delta^{(N)} \rangle &=&\int_{-\infty}^\infty t^k \ \Delta^{(N)} 
\\\nonumber &=&  \frac{1}{N} \cdot \sum_{i=1}^N (\lambda_i^{(N)})^k 
\\\nonumber &=& \frac{1}{N} \cdot \textnormal{tr }(D^k_\lambda)
\\\nonumber &=& \frac{1}{N} \cdot \textnormal{tr }((S^{-1}MS)^k)
\\\nonumber &=& \frac{1}{N} \cdot \textnormal{tr }(S^{-1}M^kS)
\\\nonumber &=& \frac{1}{N} \cdot \textnormal{tr }(SS^{-1}M^k)
\\\nonumber &=& \frac{1}{N} \cdot \textnormal{tr }(M^k) \ . 
\end{eqnarray}

\noindent Therefore, we have to discuss properties of the trace of the $k-$th power of the ensemble $M^{(N)}$. Since this is very difficult, we first show $\langle T^k , \Delta^{(N)} \rangle$ to \emph{converge in expectation}. This means

\begin{eqnarray}\label{IntroE-WertKonvergenz}
\mu_k^{(N)} :=\mathbb{E}(\langle T^k , \Delta^{(N)} \rangle)=\mathbb{E}(\frac{1}{N} \cdot \textnormal{tr }(M^k)) \overset{N\rightarrow\infty}\longrightarrow \left\lbrace \begin{array}{cc}
\sum_{\pi\in\mathcal{B}_{\frac{k}{2}}} J_\alpha(\pi) & \textnormal{for $k$ even}
\\ 0 & \textnormal{otherwise,}
\end{array}\right.
\end{eqnarray}

\noindent see Theorem \ref{TheoremMomentGleichSummeIntegrale}. The sum runs over \emph{non-crossing pair-partitions} (rooted trees) the set of which is denoted by $\mathcal{B}_{\frac{k}{2}}$. Furthermore, $J_\alpha$ is an integral connected with the weight $\alpha$, i.e.

\begin{eqnarray}\nonumber
J_\alpha(\pi):= \underbrace{\int_0^1 \cdots \int_0^1}_{(\frac{k}{2}+1)-\textnormal{times}} \prod_{\{ i,j\}: \{ g_i,g_j \}\in\mathcal{K}_\pi} \alpha^2 (|x_i-x_j|) \ dx_1 \ \ldots \ dx_{\frac{k}{2}+1} \ ,
\end{eqnarray}

\noindent see Definition \ref{DefinitionIntgralUeberBaume}. $\mathcal{K}_\pi$ denotes the edges of the \emph{$\pi$-adopted graph}, see below. Theorem \ref{TheoremNotwendigeUndHinreichendeBedfuerSCL} shows that (for even $k\in\mathbb{N}$) we \emph{essentially} have

$$ \sum_{\pi\in\mathcal{B}_{\frac{k}{2}}} J_\alpha(\pi) = C_{\frac{k}{2}} $$

\noindent if - and only if - 

$$ \varphi(x):= \int_0^1 \alpha^2(|x-y|) \ dy \equiv \varphi_0 \ . $$

\noindent Let $M:=(\beta_{ij})_{(1\leq i,j \leq N)}$ be a real-valued matrix. In order to compute the left hand side moments of (\ref{IntroE-WertKonvergenz}) we use

$$ (M^k)_{(ij)}=\sum_{r_2,...,r_k=1}^N \beta_{ir_2}\cdot \beta_{r_2r_3}\cdot \ldots \cdot \beta_{r_kj} $$

\noindent which can easily be proved by induction. This gives

\begin{eqnarray}\label{RelevanteSummeBerechnungMomentEinleitung}
\mu_k^{(N)} &=& \frac{1}{N\cdot N^{k/2}}\cdot \sum_{i_1,...,i_k=1}^N \alpha \big(\frac{|i_1-i_2|}{N} \big)\cdot \alpha \big(\frac{|i_2-i_3|}{N}\big)\cdot \ldots \cdot \alpha \big(\frac{|i_k-i_1|}{N} \big) \times 
\\\nonumber && \mathbb{E}(X_{i_1i_2} \cdot X_{i_2i_3} \cdot \ldots \cdot X_{i_ki_1}) \ . 
\end{eqnarray}

\noindent Initially, the sum runs over all indices $i\in\{ 1,...,N \}^k$. In order to reduce sum (\ref{RelevanteSummeBerechnungMomentEinleitung}) to those summands which asymptotically contribute to it, we will analyze the structure of the path $(i_1,i_2,...,i_k)$. This path is considered as a walk through a \emph{simple graph}. We will introduce a (simple) graph as a pair $(\mathcal{G};\mathcal{K})$ consisting of abstract \emph{nodes} $g_s\in\mathcal{G}$ and a set $\mathcal{K}$ of (\emph{undirected} or \emph{simple}) edges, see Definition \ref{erinnerungGraphen}. These edges will be used to control the dependence of the random variables. It turns out that only very special graphs asymptotically contribute to sum (\ref{RelevanteSummeBerechnungMomentEinleitung}). These \emph{rooted trees} will be uniquely connected to non-crossing pair-partitions via \emph{adopted sequences}, see Theorem \ref{SatzGenau1adaptierteFolge}.

Considering special weights $\alpha$, we will get also results for the eigenvalue-statistic for random \emph{band-matrices} with dependences. Example \ref{BspPeriodischeBandMatrizen} proves the SCL for \emph{periodic} band-matrices. Corollary \ref{KorollarPeriodischeBandMatrizen} shows why the SCL does not hold for band-matrices (with dependences) the \emph{band-width} of which is proportional to the dimension of the matrix. In particular results shown in \cite{BMP} for the \emph{Wigner case} are reproduced and generalized for the above ensemble. We further analyze band-matrices the band-width of which behaves like $o(N)$ and show that their moments converge against those of the SCL, see Theorem \ref{TheoremKleinOvonN}. Still an equivalence relation is used to control the dependence of the random variables. It is clear that the restrictions on this equivalence relation have to be adjusted, see (\ref{BedingungO1})-(\ref{BedingungO3}).

Finally, weak convergence in probabilty is proved to show the \emph{full} SCL. Once knowing that (for even $k\in\mathbb{N}$)

$$ \langle T^k , \mathbb{E}(\Delta^{(N))} \rangle:= \mathbb{E}(\langle T^k ,\Delta^{(N)} \rangle) \overset{N\rightarrow\infty}\longrightarrow \sum_{\pi\in\mathcal{B}_{\frac{k}{2}}} J_\alpha(\pi) \ ,  $$

\noindent one shows

\begin{eqnarray}\label{letzterSchrittIntro}
\mathbb{P}\big( |\langle T^k, \Delta^{(N)} \rangle - \langle T^k, \mathbb{E}(\Delta^{(N))} \rangle)| > \varepsilon  \big) \overset{N\rightarrow\infty}\longrightarrow 0 \ \forall \ \varepsilon>0 \textnormal{ and } k\in\mathbb{N} \ . \end{eqnarray}

\noindent This implies weak convergence in probability, see \cite{KIR2}, \cite{AGZ}. In order to prove (\ref{letzterSchrittIntro}) it is sufficient to show that

$$ \mathbb{V}\big( \langle T^k, \Delta^{(N)} \rangle \big) = \mathbb{E}\big( \langle T^k, \Delta^{(N)} \rangle^2 \big) - \mathbb{E}\big( \langle T^k, \Delta^{(N)} \rangle \big)^2 \overset{N\rightarrow\infty}\longrightarrow 0 \ \forall \ k\in\mathbb{N} \ .  $$

\noindent The key observation to this fact is \textsc{Chebyshev}'s inequality which states

$$ \mathbb{P}(|X-\mathbb{E}(X)|>\varepsilon) \leq \frac{\mathbb{V}(X)}{\varepsilon^2} \ \forall \ \varepsilon>0 \ .  $$

This paper is concluded with the proof for weak convergence in probability not only for the ensemble $M^{(N)}$ but also for band ensembles, see Theorem \ref{TheoremSchwacheKonvergenzinWkeitVollerFall} and Corollary \ref{KorollarSchwacheKonvInWkeitOvonN}. In order to do so the variance $\mathbb{V}( \langle T^k, \Delta^{(N)} \rangle)$ is estimated using methods from Chapter \ref{section:MomentsOfRandomMatrices}.

\bigskip

Vorrei ringraziare - al termine di quattro anni di dottorato di ricerca - tutte le persone, specialmente Prof. Dr. \textsc{Werner Kirsch}, che a vario titolo mi hanno accompagnato e senza le quali aiuto questa tesi non sarebbe stata possibile realizzare. My very special thanks go to Prof. Kirsch who acted as my supervisor during the last four years of my research and who supported me in an outstanding way. Prof. Kirsch never got tired to answer my numerous questions - thank you !

\newpage

\section{Graphs, Trees and Partitions of a Set}\label{section:Multigraphs}

\noindent As explained in the introduction certain \emph{tree-like} graphs play a key role in computing moments of a random matrix. Therefore in this section we will prove some basic properties of \emph{trees} and associated sequences of their nodes and show the connection between non-crossing pair-partitions of a set and certain ordered (\emph{rooted}) trees.

\begin{definition}[Graphs, Paths, Trees and Colour]\label{erinnerungGraphen}

	\begin{enumerate}
	
\item[]

\item[(i)] Let $\mathcal{G}$ be a countable set and let $\mathcal{T}_2(\mathcal{G})$ denote the set of all subsets of $\mathcal{G}$ consisting of exactly two elements. Consider $\mathcal{K}\subset\mathcal{T}_2(\mathcal{G})$. Then the pair $(\mathcal{G},\mathcal{K})$ is called a \emph{simple graph} (or also \emph{undirected graph}) or simply a \emph{graph}.


	

\item[(ii)]  A \emph{path} (of length $k$) in a graph is a sequence $(g_1, ..., g_k)$ with  

$$\{g_i, g_{i+1}\}\in\mathcal{K} \ \forall \ i\in\{ 1,...,k-1 \} \ . $$
	
\item[(iii)] A path is called a \emph{circle} if (ii) holds and additionally one has $\{g_k,g_1\}\in\mathcal{K}$.

\item[(iv)] $(\mathcal{G};\mathcal{K})$ is called \emph{connected}, if for all nodes $i,j \in \mathcal{G}$ there is a path with
	
	$$ (i=g_1, ... , j=g_k) \ . $$

\item[(v)]	A simple graph is called a \emph{tree}, if it is connected and for every $i\in\mathcal{G}$ there is no path $(g_1,...,g_k)$ with $g_1=i$, $\{g_k;g_1\}\in\mathcal{K}$ and 

$$ \{ g_r, g_{r+1} \} \neq \{ g_s, g_{s+1} \} \ \forall \ r\neq s \ ,  $$

\noindent where $k+1:=1$ is defined cyclically.

\item[(vi)] Every injective mapping $c:\mathcal{G}\rightarrow \mathbb{N}$ is called a \emph{colouring} of the graph. A \emph{labeled} (or \emph{coloured}) graph is a triple $(\mathcal{G};\mathcal{K}; c)$, where $c$ is a colour of the graph.
	
	\end{enumerate}

\end{definition}

\begin{remark}
Consider a simple graph $(\mathcal{G},\mathcal{K})$. The set $\mathcal{K}$ is interpreted as the set of edges of the graph. An element $\{ g_l,g_m \}\in\mathcal{K}$ is the connection between the nodes $g_l$ and $g_m$. One can not distinguish weather the connection goes \emph{from} $g_l$ \emph{to} $g_m$ or vice versa. This is the reason why the graph is called \emph{undirected}.

\end{remark}

Next we will show that certain partitions of a set (so called \emph{non-crossing pair-partitions}) can be used to describe trees.

\begin{definition}[Partition and Pair-Partition]\label{PartitionUndPaarPartition}
Let $S\neq\emptyset$ be a finite set and let $2^S$ denote the powerset of $S$. A subset $\pi:= \{ B_1, ..., B_r \}\subset 2^S$ of the powerset is called \emph{partition} of $S$, if

$$ S = \overset{\bullet}\bigcup_{\rho\in \{ 1, ..., r \}} B_\rho \ .  $$

\noindent A partition is called a \emph{pair-partition}, if $|B_\rho|=2$ for all $\rho\in\{ 1, ..., r \}$. We denote the set of all pair partitions of $S$ by $\mathcal{P}_2(S)$. Pair partitions can only occur on sets the cardinality of which is even.
\end{definition}

\begin{remark}
Let $S$ be a set of cardinality $k$. Let $\sim$ be an equivalence relation on $S$. Then this equivalence relation gives rise to a partition $\pi$ of $\{ 1, ..., k \}$ as follows: If we have $S=\{ s_1, ..., s_k \}$ then we set

$$B_{m}:= \big\{  j\in\{ 1,...,k \} | \ \ s_m\sim s_j \big\} \ . $$

\noindent For two \emph{blocks} $B_q$ and $B_r$ we have either $B_q=B_r$ or $B_q\cap B_r = \emptyset$, as can easily be seen: Suppose that $B_q\cap B_r \neq \emptyset$. Then we can take an element $x\in B_q\cap B_r$. We get


$$ s_x \sim s_q \textnormal{ and } s_x \sim s_r \Longrightarrow s_q \sim s_r \Longrightarrow B_q=B_r \ .  $$

\noindent We can now define a partition

 $$\pi:= \{ B_{1},...,B_{k} \} =: \{ C_1,...,C_r \} $$

\noindent with blocks $C_\rho$. Any of the (disjoint) blocks $C_\rho$ is the congruence class of a certain $m\in\{ 1,...,k \}$.
\end{remark}

\begin{definition}[Crossing Pair-Partitions]

Let $k\in\mathbb{N}$ be a natural number and let $\pi:= \{ B_1, ..., B_r \}$ be a partition of $\{ 1,...,k \}$.

\begin{enumerate}


	\item[(i)]  A pair-partition $\pi$ is called \emph{crossing}, if
	
	$$ \exists \ 1\leq a<b<c<d \leq k : \{ a,c \} \in\pi \ni \{ b,d \} \ . $$

	\noindent $\pi$ is called \emph{non-crossing}, if it is not crossing.	
	
	\item[(ii)] We define
	
	$$\mathcal{B}_{\frac{k}{2}}:= \left\lbrace \pi | \pi \textnormal{ is a non-crossing pair-partition on } \{1, ..., k\} \right\rbrace  \ . $$
	
	\end{enumerate}
	
\end{definition}	
	
The next lemma provides a basic property of non-crossing pair-partitions.

\begin{lemma}\label{BaumeHabenBlatter}
Let $k$ be even and $\pi\in\mathcal{B}_{\frac{k}{2}}$. If we cyclically define $k+1:=1$ then there exists (at least) one $m\in\{ 1, ..., k \}$ with $\{ m, m+1 \}\in\pi$.
\end{lemma}

\begin{proof}
If the assumption is wrong then choose a block $\{ m, m+l \}\in\pi$ with a minimal $l>1$. We have

$$ \# \{m+1, ..., m+l-1 \} = l-1 > 0 \ . $$


\noindent None of the indices in $M=\{m+1, ..., m+l-1 \}$ are associated with an index outside of this set since $\pi$ is non-crossing. That contradicts with the minimality of $l$.

\end{proof}

\begin{definition}[Adopted Sequences] Let $k\in\mathbb{N}$ be a natural number and let $\pi:= \{ B_1, ..., B_r \}$ be a partition of $\{ 1,...,k \}$.
	
	\begin{enumerate}

	\item[(i)] \begin{eqnarray}\nonumber
\pi \setminus^\bullet \{ m, m+1 \} := \left\lbrace \begin{array}{ll}
\pi\setminus \{ m, m+1 \} \textnormal{, } a\mapsto a-2 \textnormal{ for } a>m & , \{ m, m+1 \}\in\pi
\\ \pi & , \{ m, m+1 \}\not\in\pi
\end{array} \right. 
	\end{eqnarray}

Here, all numbers $\{ 1,...,k \} \ni a>m$ are \emph{relabeled} to $a-2$, if a block $\{ m,m+1 \}\in\pi$ is removed from $\pi$. In this case one obviously has $\pi \setminus^\bullet \{ m, m+1 \}\in\mathcal{B}_{\frac{k}{2}-1}$.

	\item[(ii)] Let $\mathcal{G}$ be a discrete set and $\pi\in\mathcal{B}_{\frac{k}{2}}$. A sequence $g:=(g_1,...,g_k)\in\mathcal{G}^k$ is called \emph{$\pi-$adopted}, if the following is valid:
	
	\begin{enumerate}
		\item[1.] For $k=2$ the only non-crossing pair-partition is $\pi=\{ \{1,2\} \}$. The only $\pi-$adopted sequences are $(g,h)$ with $g\neq h$.		
		
		\end{enumerate}
		
For all blocks $\{ m,m+1 \}\in\pi$ we have	
		
		\begin{enumerate}
		
		\item[2.] $g_{m}=g_{m+2}$ and $g_{m+1}\neq g_s$ for all $s\neq m+1$.		
		
		\item[3.] $(g_1, ..., g_m, \overset{\wedge}g_{m+1}, \overset{\wedge}g_{m+2}, g_{m+3}, ..., g_k)$ is $\pi
		 \setminus^\bullet \{ m, m+1 \}-$adopted.

	\end{enumerate}

(The hat over elements means to remove them from the sequence.)

	\end{enumerate}

\end{definition}

\begin{example}
Consider $\pi:= \{ \{1,4\}, \{2,3\} \} \in\mathcal{B}_2$. Let $\mathcal{G}$ be a discrete set of nodes and consider

$$ g_1,g_2,g_3 \in\mathcal{G} \ \textnormal{ with } g_i\neq g_j \textnormal{ for } i\neq j \ .  $$

\noindent An adopted sequence is

$$ (g_1,g_2,g_3,g_2)  $$

\noindent as can easily be checked. This sequence can be interpreted as a walk through a (simple) graph. The corresponding nodes are $g_1,g_2$ and $g_3$. Every two consecutive nodes form an edge, which means that the set of edges is

$$ \mathcal{K}= \big\{ \{ g_1, g_2\}, \{ g_2, g_3 \} \big\} \ .  $$

\noindent On the other hand, given a graph, \emph{a walk} through its nodes is not unique. Adopted sequences - the existence and uniqueness of which have to be shown - define a certain walk (through a graph) which is unique. Thus, they can be used to describe a \emph{rooted} tree. This is a tree, where the walk through its nodes is important. We will indeed show later on, that for every $\pi\in\mathcal{B}_{\frac{k}{2}}$ there is one - and only one- adopted sequence. Therefore, a non-crossing pair-partition can be used to describe rooted trees, which turn out to determine exactly those summands, which asymptotically contribute to sum (\ref{RelevanteSummeBerechnungMomentEinleitung}). We will discuss later on, how $\pi-$adopted sequences are obtained from $\pi\in\mathcal{B}_{\frac{k}{2}}$.

\end{example}

\begin{remark}
Wigner's original proof for symmetric random matrices with independent random variables (\emph{Wigner case}) \cite{WIG} shows that there are only a \emph{few} summands which essentially contribute to sum (\ref{RelevanteSummeBerechnungMomentEinleitung}). Every index is the colour of a certain node. The graph, defined by the nodes and the run through the graph, which defines its edges, comes from non-crossing pair-partitions. This walk is unique, as we will now show. The mentioned result also remains true for the case discussed in this paper: here, certain correlations of the random variables are permitted.
\end{remark}

\begin{lemma}\label{BlockegebenHinundZuruck} Let $\pi$ be a non-crossing pair-partition on $\{ 1,...,k \}$ and let $g:=(g_1,...,g_k)$ be a $\pi-$adopted sequence. If $\{ m, m+l \}\in\pi$ for some $l>1$, then

\begin{eqnarray}\nonumber
\left\lbrace \begin{array}{ll}
g_m & = g_{m+l+1}
\\ g_{m+1} & = g_{m+l}  
\end{array} \right. 
	\end{eqnarray}

\end{lemma}

\begin{proof} Since $\pi$ is non-crossing any of the $l-1>0$ indices in the set 

$$S:=\{ m+1,...,m+l-1 \}$$

\noindent forms a block with another index of the same set. In particular, the number $l-1$ is even. We therefore get $\frac{l-1}{2}$ blocks between $m$ and $m+l$ from which we remove $\frac{l-3}{2}$ according to Lemma \ref{BaumeHabenBlatter}. We obtain an element $\tilde{\pi}\in\mathcal{B}_{\frac{k-(l-3)}{2}}$ with 

$$ \{ m,m+3 \}, \{ m+1,m+2 \} \in\tilde{\pi} \ . $$

\noindent By definition the sequence 

$$ \tilde{g}:= (g_1,..., g_{m+1},\overset{\wedge}g_{m+2},...,\overset{\wedge}g_{m+l-2},g_{m+l-1},...,g_k) =: (\tilde{g}_1,...,\tilde{g}_{k-(l-3)}) $$

\noindent is $\tilde{\pi}-$adopted, which implies 

$$ g_{m+1}=\tilde{g}_{m+1} = \tilde{g}_{m+3}=g_{m+l} \ . $$

\noindent Again by definition the sequence 

$$ \tilde{\tilde{g}} := (g_1,..., g_{m+1},\overset{\wedge}g_{m+2},...,\overset{\wedge}g_{m+l},g_{m+l+1},...,g_k) =: (\tilde{\tilde{g}}_1,...,\tilde{\tilde{g}}_{k-(l-1)}) $$

\noindent is $\tilde{\tilde{\pi}}:=\tilde{\pi}\setminus^\bullet \{ m+1,m+2 \}-$adopted, which implies

$$ g_{m}=\tilde{\tilde{g}}_{m} = \tilde{\tilde{g}}_{m+2}=g_{m+l+1}   $$

\noindent since $\{ m,m+1 \}\in\tilde{\tilde{\pi}}$.

\end{proof}

Prior to be able to prove a connection between adopted sequences, graphs and non-crossing pair-partitions, we need to define what we mean by an isomorphism between two sequences.

\begin{definition}[Isomorphism between sequences]
Let $\mathcal{G}$ be a countable set of (abstract) nodes. Let 

$$ g:=(g_1,...,g_k) \in \mathcal{G}^k \ni (h_1,...,h_k)=:h $$

\noindent be two sequences of (ordered) nodes. The sequences $g$ and $h$ are called \emph{equivalent}, if there is a bijective mapping $\sigma:\mathcal{G}\rightarrow\mathcal{G}$ with

$$ \sigma(g_i) = h_i \ \forall \ i\in \{ 1,...,k \} \ .  $$



	


\noindent If $g$ and $h$ are equivalent we also say that they are the same sequences \emph{up to an isomorphism}.

\end{definition}

\begin{theorem}\label{SatzGenau1adaptierteFolge}
Let $k$ be even and let $\pi\in\mathcal{B}_{\frac{k}{2}}$. Then there exists (up to an isomorphism) exactly one $\pi-$adopted sequence $(G_1,...,G_k)=:G(\pi)$. Further one has $\# \{ G_1, ..., G_k \}=\frac{k}{2}+1$.

\end{theorem}

\begin{proof}
We will prove the Theorem via induction. For $k=2$ the only $\pi-$adopted sequence is $(g_1,g_2)$, $g_1\neq g_2\in\mathcal{G}$ by definition. This also shows the existence of an adopted sequence in this case. We discuss the induction-step next.

\begin{description}

	\item[Existence:] Let $\pi\in\mathcal{B}_{\frac{k}{2}}$ be a non-crossing pair-partition. According to Lemma \ref{BaumeHabenBlatter} we can choose a block $\{ m,m+1 \}\in\pi$. By hypothesis there \emph{exists} a $\pi\setminus^\bullet\{ m,m+1 \} =: \tilde{\pi}-$adopted sequence
	
	$$ \tilde{G}:= (g_1, ..., g_m, g_{m+1},...,g_{k-2} ) $$

\noindent with \emph{some} nodes $g_1,...,g_{k-2}\in\mathcal{G}$. We choose an element $h\in \mathcal{G}\setminus \{ g_1, ..., g_{k-2} \}$ and define 

$$ G:= (g_1, ..., g_m, h, g_m, g_{m+1},...,g_{k-2} ) \in\mathcal{G}^k \ . $$

\noindent Then $G$ is $\pi-adopted$ by construction.

	\item[Uniqueness:] We now choose \emph{the left most} block of $\pi\in\mathcal{B}_{\frac{k}{2}}$, i.e. a block $\{ m,m+1\}\in\pi$ with a \emph{minimal} index $m\in \{ 1,...,k \}$. If $G=(G_1,...,G_k)$ is $\pi-adopted$, by definition, we necessarily obtain $G_m=G_{m+2}$ and $G_{m+1}\neq G_s$ for all $m+1\neq s$. By definition the sequence
	
	$$ \tilde{G}:= (G_1, ..., G_m, \overset{\wedge}G_{m+1}, \overset{\wedge}G_{m+2}, G_{m+3}, ..., G_k)  $$

\noindent is $\pi\setminus^\bullet\{ m,m+1 \} =: \tilde{\pi}-$adopted and by hypothesis there is (up to an isomorphism) only one $\tilde{\pi}-$adopted sequence, let us say

\begin{eqnarray}\nonumber
\tilde{G}= (g_1,...,\underbrace{g}_{\textnormal{position $m$}},..., g_r) \ .
\end{eqnarray}

\noindent By comparision we obtain that $G_m=g=G_{m+2}$ and $G_{m+1}:=h\in \mathcal{G}\setminus \{ g_1, ..., g, ... g_r \}$ is a new node by definition. This implies the uniqueness of a $\pi-$adopted sequence $G$ since the index $m$ was minimal.

\end{description}

\noindent Finally by hypothesis one has $\# \{ g_1, ..., g, ... g_r  \}=\frac{k-2}{2}+1$ and therefore 

$$ \# \{ G_1, ..., G_k \} = \frac{k-2}{2}+1 + 1 = \frac{k}{2}+1 \ . $$









\end{proof}

\begin{definition}(Adopted Graph)\label{DefinitionAdaptierterGraph}
Let $k$ be even and let $\pi\in\mathcal{B}_{\frac{k}{2}}$. Let $G(\pi):=(G_1,...,G_k)$ be the only $\pi-$adopted sequence. We set

\begin{eqnarray}\nonumber
\mathcal{G}_\pi &:=&\{ G_1, ..., G_k \}=:\{ g_1, ..., g_{\frac{k}{2}+1} \} \ ,
\\\nonumber \mathcal{K}_\pi &:=& \big\{ \{ G_m, G_{m+1} \} | \ m=1,...,k \big\} \ ,
\end{eqnarray}

\noindent The pair $(\mathcal{G}_\pi; \mathcal{K})$ is called \emph{$\pi-$adopted graph}. It has $\frac{k}{2}=|\mathcal{K}|$ edges, $\frac{k}{2}+1$ different nodes. It is unique (up to an isomorphism) according to Theorem \ref{SatzGenau1adaptierteFolge}.

\end{definition}

\newpage

\section{Moments of Random Matrices}\label{section:MomentsOfRandomMatrices}

In this section we will prove the main result explained in the introduction. We will set up the precise setting and develop techniques to compute the moments of certain random matrices which generalize known results, see  \cite{HSBS}, \cite{BMP}, \cite{WIG}, \cite{WIG2}.


Let $\alpha:[0;1]\rightarrow\mathbb{R}$ be a bounded, Riemann-integrable function. For every $N\in\mathbb{N}$ we consider symmetric random matrices

\begin{eqnarray}\label{Zufallsmatrix}
M^{(N)} := \frac{1}{\sqrt{N}} \left( \alpha\big(\frac{|i-j|}{N} \big) \cdot X^{(N)}_{ij} \right)_{(1\leq i,j \leq N)} \ . 
\end{eqnarray}

\noindent with random variables $X^{(N)}_{ij}\equiv X^{(N)}_{ji}$ as explained in the introduction. The random variables do not have to be independent, nevertheless restrictions are required. First we want the family $X^{(N)}_{ij}$ to be centered and to be normalized, which means

$$ \mathbb{E}\big(X^{(N)}_{ij} \big)=0 \textnormal{ and } \mathbb{V}\big(X^{(N)}_{ij}\big) = 1 \ \qquad \forall \ N\in\mathbb{N}  $$

\noindent and since we will compute moments we surely require all moments to exist and to be bounded, which means

$$ \sup_N \max_{i,j} \mathbb{E}\big(|X^{(N)}_{ij}|^k\big) \leq R_k < +\infty \ \qquad \forall \ k\geq 1 \ . $$

\noindent This assumption is necessary for certain upper bounds in computing moments of $M^{(N)}$. In order to control the dependence of the random variables $X_{ij}:=X^{(N)}_{ij}\not\equiv 0$ we introduce an equivalence relation $\sim$ on $\{ 1,...,N \}^2$. Whenever $(p,q)\not\sim(r,s)$ we assume $X_{pq}$ and $X_{rs}$ to be independent. Since we require $M:=M^{(N)}$ to be symmetric, we only consider equivalence relations with $(p,q)\sim (q,p)$ for all $p,q\in \{ 1,...,N \}$.

From now on we only consider equivalence relations which satisfy the following conditions.

\begin{eqnarray}\label{AequivalenzrelationBedingung1}
\max_{p} &\#& \{ (q,r,s)\in\{ 1,...,N\}^3 | (p,q)\sim (r,s) \} = o(N^2)
\\ \label{AequivalenzrelationBedingung2} \max_{p,q,r} &\#& \{ s\in\{ 1,...,N \} | (p,q)\sim(r,s) \} \leq B < +\infty
\\ \label{AequivalenzrelationBedingung3}  &\#& \{ (p,q,r)\in\{ 1,...,N \}^3 | (p,q)\sim(q,r) \textnormal{ and } r\neq p \} = o(N^2) \ . 
\end{eqnarray}

\smallskip

\begin{remark}
Equivalence relations with restrictions (\ref{AequivalenzrelationBedingung1})-(\ref{AequivalenzrelationBedingung3}) were also considered by \cite{HSBS}.
\end{remark}

\smallskip

\begin{remark}
Let us consider an equivalence relation with the only congruence classes



$$ [(p,q)]= \{ (p,q), (q,p) \} \textnormal{ for all } p,q\in \{ 1,...,N \} \ . $$

\noindent This is exactly the Wigner case: The random variables $X_{ij}$ and $X_{pq}$ are independent as long as $(i,j)\not\in \{ (p,q), (q,p) \}$.

We will verify, that conditions (\ref{AequivalenzrelationBedingung1})-(\ref{AequivalenzrelationBedingung3}) are fullfilled.

\begin{description}
	\item[Condition (\ref{AequivalenzrelationBedingung1}):] Let $p\in\{ 1,...,N\}$ be fixed. One has $N$ possibilities to choose $q$ but then the pair $(r,s)$ is fixed from a set of two elements, which gives most $2N$ choices.
	
	\item[Condition (\ref{AequivalenzrelationBedingung2}):] Let $(p,q,r)\in\{ 1,...,N\}^3$ be fixed. Then we have at most one possibility to choose $s$ if $r\in\{ p,q \}$. Otherwise we have none.
	
	\item[Condition (\ref{AequivalenzrelationBedingung3}):] We consider two pairs $(p,q)\sim(q,r)$. Then necessarily $p=r$ which is forbidden in the above condition. The set $\{ (p,q)\sim(q,r) \textnormal{ and } r\neq p \}$ is empty.

\end{description}

\noindent This shows that all results obtained in this paper also hold for the Wigner case. 

\end{remark}

\bigskip

We will now introduce the abstract setting to calculate the limit of sum (\ref{RelevanteSummeBerechnungMomentEinleitung}). Let $2^{\{ 1,...,k \}}$ denote the powerset of $\{ 1,...,k \}$ and let $\mathcal{G}$ be a countable set (of abstract nodes). Let $c: \{ g_1,...,g_k \}\rightarrow \{ 1,...,N \}$ be a colour of the nodes $g_\kappa\in\mathcal{G}$. In particular, $c$ is injective by definition. To any sequence $(g_1,...,g_k)\in\mathcal{G}^k$ we associate \emph{blocks}

$$ B_{m}:=  \big\{ j\in\{ 1,...,k \} \ | \ \big(c(g_m),c(g_{m+1})\big)\sim \big(c(g_j),c(g_{j+1})\big)   \big\} \ , \ m\in\{  1,...,k\} \ , $$

\noindent where $k+1:=1$ is defined cyclically. These blocks give raise to a partition 

$$\pi=\pi(c):= \{ B_{1},...,B_{k} \} =: \{ C_1,..., C_r \} \ , $$ 

\noindent where $(C_\rho)_{\rho=1,...,r}$ is a family of blocks with $C_a\neq C_b$ for $a\neq b$. Since $\sim$ is an equivalence relation, $\pi(c)=\{ C_1,...,C_r \}$ is a partition of $\{ 1,...,k \}$, as the remark after Definition \ref{PartitionUndPaarPartition} shows. The pair $\gamma:=\big((g_1,...,g_k),c \big)$ is referred to as a \emph{cycle}. The above discussion shows that we can obtain any cycle from a partition via $\mathfrak{f}^{-1}$, where $\mathfrak{f}$ denotes the mapping

\begin{eqnarray}\nonumber
\mathfrak{f}: \big\{ \gamma:=\big((g_1,...,g_k),c \big) \ | \ \gamma \textnormal{ is a cycle} \big\} &\rightarrow& \mathcal{P}(k):= \{ \pi\subset 2^{\{ 1,...,k \}} | \pi \textnormal{ is a partition of } \{ 1,...,k \} \}
\\\nonumber \big((g_1,...,g_k),c \big) &\mapsto& \pi(c) \ . 
\end{eqnarray}


\noindent We define

\begin{eqnarray}\nonumber
\mathcal{E}^{(N)}_k(\pi) := \big\{ \Gamma:= \big((g_1,...,g_k),\pi,c \big) | \{ l,m \}\subset B_i\in\pi &\Leftrightarrow& \\\nonumber \big(c(g_m),c(g_{m+1})\big)&\sim& \big(c(g_l),c(g_{l+1})\big) \ , 
\\\nonumber  c: \{ g_1,...,g_k \} &\hookrightarrow& \{ 1,...,N \} , \ g_s\in\mathcal{G}  \big\} \ ,
\end{eqnarray}

\noindent where '$\hookrightarrow$' underlines, that the mapping $c$ is injective. For $\Gamma\in\mathcal{E}^{(N)}_k:=\mathcal{E}^{(N)}_k(\pi)$ we set $c_j:=c(g_j)$ and cyclically define

$$ X_N(\Gamma):= \prod_{j=1}^k X_{c_j c_{j+1}} \textnormal{ and } \alpha(\Gamma) := \prod_{j=1}^k \alpha \big(\frac{|c_j-c_{j+1}|}{N} \big) \ . $$

\noindent Since every cycle $\gamma$ can be obtained by a partition $\pi\in\mathcal{P}(k)$, the limit of sum (\ref{RelevanteSummeBerechnungMomentEinleitung}) can be rewritten as

\begin{eqnarray}\label{MomentSummeUeberE}
\mu_k := \lim_{N\rightarrow\infty} \mu_k^{(N)} = \lim_{N\rightarrow\infty} \frac{1}{N^{1+k/2}} \sum_{\pi\in\mathcal{P}(k)} \ \sum_{\Gamma\in\mathcal{E}^{(N)}_k(\pi)} \alpha(\Gamma) \cdot \mathbb{E}(X_N(\Gamma)) \ .
\end{eqnarray}


In order to work out the summands which asymptotically contribute to sum (\ref{MomentSummeUeberE}) we introduce the following projector.

\begin{eqnarray}\label{ProjektorP}
P: \mathcal{E}^{(N)}_k &\longrightarrow& \mathcal{P}(k)
\\\nonumber \Gamma &\longmapsto& \pi \ .
\end{eqnarray}

\begin{lemma}[Upper bound]\label{ObereAbschatzung}
Let $\pi:= \{ B_1,...,B_r \}\in\mathcal{P}(k)$ be a partition. Then 

$$ \# P^{-1}(\pi) \leq  N^{r+1} \cdot B^{k-r-1} \ , $$

\noindent where $B$ denotes the upper bound from condition (\ref{AequivalenzrelationBedingung2}).

\end{lemma}

\begin{proof}
Without loss of generality we can assume that $1\in B_1$ and we colour $g_1$ and $g_2$. This gives at most $N^2$ possibilities.





We colour the nodes $g_3,...,g_k$ \emph{successively} and note that for an index $l\in \{ 2,...,k \}$ there are two possibilities:

\begin{enumerate}

	\item[(i)] $\exists \ i \in \{ 1,...,r \}  \ \exists \ m<l: \{ m,l \}\subset B_i$. Then we have
	
	$$ (c(g_m),c(g_{m+1})) \sim (c(g_l),c(g_{l+1})) $$

\noindent and according to condition (\ref{AequivalenzrelationBedingung2}) we have at most $B$ choices to colour node $g_{l+1}$.

	\item[(ii)] 	$\forall \ m<l: \{ m,l \}\not\subset B_i$. Then we can colour node $g_{l+1}$ with some numbers of $\{ 1,...,N \}$. This gives less than $N$ values due to injectivity.
	
\end{enumerate}

\noindent Block $B_1$ was used to colour the first two nodes. Because there are $r-1$ blocks left in the partition, we can freely colour $r-1$ nodes. The remaining $k-2-(r-1)=k-r-1$ indices are constrained by condition (\ref{AequivalenzrelationBedingung2}). We get

$$ \# P^{-1}(\pi) \leq N^2 \cdot N^{r-1} \cdot B^{k-r-1} = N^{r+1} \cdot B^{k-r-1} \ . $$

\end{proof}

\begin{proposition}\label{KorollarNurPaarPartitionenBeitrag}
Let $\pi:=\{ B_1,..., B_r \}\in\mathcal{P}(k)$ be a partition with $r\neq\frac{k}{2}$. Then we have

$$ \lim_{N\rightarrow\infty} \frac{1}{N^{1+k/2}} \sum_{\Gamma\in\mathcal{E}^{(N)}_k(\pi)} \alpha(\Gamma) \cdot \mathbb{E}(X_N(\Gamma)) = 0 \ . $$

\end{proposition}

\begin{proof}
First suppose $r>\frac{k}{2}$. Then there is at least one block which consists of exactly one element. Otherwise we would have more than $2r > k$ elements in $\{ 1,...,k \}$. Hence we have a pair $(i_j,i_{j+1}):=(c(g_j),c(g_{j+1}))$ which appears only once. Furthermore, $\alpha$ is bounded, which implies the existence of a ($k-$dependent) constant $A=A(k)<+\infty$ with 
	
	$$ \alpha(\Gamma) \leq A \ \qquad  \forall \ \Gamma\in\mathcal{E}_k(\pi) \ . $$

\noindent For $\Gamma\in\mathcal{E}_k(\pi)$ we get

$$ |\alpha(\Gamma) \cdot \mathbb{E}(\Gamma) | \leq A \cdot  \prod_{s\neq j} \mathbb{E}(X_{i_s i_{s+1}}) \cdot \mathbb{E}(X_{i_j i_{j+1}}) = 0 \ . $$

Now we suppose that $r<\frac{k}{2}$. \textsc{H\"older's} inequality implies 

$$ |\mathbb{E}(X_N(\Gamma))| \leq R_k \ \qquad \forall \ \Gamma\in\mathcal{E}_k(\pi) \ . $$

\noindent Using Lemma \ref{ObereAbschatzung} we see that

\begin{eqnarray}\nonumber
\big| \frac{1}{N^{1+k/2}} \sum_{\Gamma\in\mathcal{E}^{(N)}_k(\pi)} \alpha(\Gamma) \cdot \mathbb{E}(X_N(\Gamma)) \big|
&\leq& A \cdot R_k \cdot \frac{1}{N^{1+k/2}} \cdot \# \mathcal{E}^{(N)}_k(\pi) 
\\\nonumber &=&  A \cdot R_k \cdot \frac{1}{N^{1+k/2}} \cdot \# P^{-1}(\pi)
\\\nonumber &\leq& A \cdot R_k \cdot N^{r+1-1-k/2} \cdot B^{k-r-1} 
\\\nonumber &=& A \cdot R_k \cdot N^{r-k/2} \cdot B^{k-r-1} 
\\\nonumber &\overset{N\rightarrow\infty}\longrightarrow& 0 \ ,
\end{eqnarray}

\noindent which completes the proof.

\end{proof}

\begin{remark} 
Proposition \ref{KorollarNurPaarPartitionenBeitrag} shows that only partitions with exactly $k/2$ blocks give a contribution to sum (\ref{MomentSummeUeberE}). Its proof also shows that only such partitions give an asymptotically non-vanishing contribution to sum (\ref{MomentSummeUeberE}) the blocks of which consist of at least two elements. These two facts together imply that only \emph{pair-partitions} contribute to sum (\ref{MomentSummeUeberE}). Next we will show that only \emph{non-crossing} pair-partitions give a non-vanishing contribution to the mentioned sum. Later on we will see that only very special sequences $(g_1,...,g_k)$, which are in connection with $\Gamma\in\mathcal{E}_k$, give the non-vanishing summands of (\ref{MomentSummeUeberE}). These will turn out to be the $\pi-$adopted sequences.

\end{remark}

\begin{proposition}\label{BlaetterEntfernungsProposition}
Let $\pi\in\mathcal{P}_2(\{ 1,...,k \})$ be a pair-partition and suppose that there is a block $\{ m,m+1 \}\in\pi$. Then 

$$ \# P^{-1}(\pi) \leq N\cdot \# P^{-1}(\tilde{\pi}) + o(N^{1+k/2}) \ . $$

\noindent where $\tilde{\pi}:=\pi\backslash^\bullet\{ m,m+1 \}$.

\end{proposition}

\begin{proof}
For an element $\Gamma:= \big((g_1,...,g_k),\pi,c \big)\in P^{-1}(\pi)$ we have a look at the nodes $g_m,g_{m+1}$ and $g_{m+2}$. 

\begin{description}

	\item[Case (i):] $g_m=g_{m+2}$. Then $\tilde{\Gamma}:=\big((g_1,...,g_m,\overset{\wedge}g_{m+1},\overset{\wedge}g_{m+2},g_{m+3},...,g_k),\tilde{\pi},c \big)\in P^{-1}(\tilde{\pi})$ and there are at most $N$ choices to colour node $g_{m+1}$. This shows
	
	$$  \# P^{-1}(\pi) \leq N\cdot \# P^{-1}(\tilde{\pi}) \ . $$

	\item[Case (ii):] $g_m\neq g_{m+2}$. According to condition (\ref{AequivalenzrelationBedingung3}) there are $o(N^2)$ possibilities to colour the nodes $g_m,g_{m+1}$ and $g_{m+2}$. Analogously to the proof of Lemma \ref{ObereAbschatzung} we see that if node $g_{l+1}$ has to be coloured there are two possibilities:
	
	\begin{enumerate}
	
		\item[1)] $\exists \ m+1<s<l: \{ s,l \}\in\pi$. Then $c(g_{l+1})$ is constrained by condition (\ref{AequivalenzrelationBedingung2}). This gives at most $B$ values, since $g_s,g_{s+1}$ and $g_l$ are already coloured as in the proof of Lemma \ref{ObereAbschatzung}.
		
		\item[2)] 	$\forall \ m+1<s<l: \{ s,l \}\not\in\pi$. Then $c(g_{l+1})$ is not constrained and can take at most $N$ values.
	
	\end{enumerate}	

\noindent Since $\tilde{\pi}$ has $k/2-1$ blocks there are $N^{k/2-1}$ possibilities to freely colour $k/2-1$ nodes. The remaining 

$$k- \big(\frac{k}{2}-1\big)-3=\frac{k}{2}-2$$

\noindent nodes are constrained by condition (\ref{AequivalenzrelationBedingung2}), which gives less than $B^{k/2-2}$ possibilities. This shows

$$ \# P^{-1}(\pi) \leq o(N^2) \cdot N^{k/2-1} \cdot B^{k/2-2} = o(N^{1+k/2}) \ , $$
	
	\noindent since $B$ is assumed to be independent of $N$.
\end{description}
\end{proof}

\begin{remark}
Consider the Wigner case. If there is a block $\{m,m+1  \}\in\pi$, then $g_m\neq g_{m+2}$ is not allowed. Therefore Proposition \ref{BlaetterEntfernungsProposition} reduces to case (i) which states 

$$ \# P^{-1}(\pi) \leq N\cdot \# P^{-1}(\tilde{\pi}) \ . $$

\end{remark}

In order to be able to give an upper estimate for non-adopted sequences later on, we will prove the following Corollary.

\begin{corollary}\label{BlaetterEntfernungsKorollarFurNichtadaptierteSequenzen}
Let $\pi\in\mathcal{B}_{\frac{k}{2}}$ be a non-crossing pair-partition. Let $\overline{E_k}$ denote the subset

$$ \mathcal{E}^{(N)}_k \supset \overline{E_k} := \big\{ \Gamma\in\mathcal{E}^{(N)}_k | \ (g_1,...,g_k) \textnormal{ is not  $\pi-$adopted} \big\} \ .  $$

\noindent We define $\Pi:= P|_{\overline{E_k}}:\overline{E_k}\longrightarrow\mathcal{B}_{\frac{k}{2}}$, where $P$ denotes the projector (\ref{ProjektorP}). If there is a block $\{ m,m+1 \}\in\pi$, then

$$ \# \Pi^{-1}(\pi) \leq N\cdot \# \Pi^{-1}(\tilde{\pi}) + o(N^{1+k/2}) $$

\noindent where $\tilde{\pi}:=\pi\backslash^\bullet\{ m,m+1 \}\in\mathcal{B}_{\frac{k-2}{2}}$. 

\end{corollary}

\begin{proof}
We take an element $\Gamma:= \big((g_1,...,g_k),\pi,c \big)\in \Pi^{-1}(\pi)$ and we have a look at the nodes $g_m,g_{m+1}$ and $g_{m+2}$. 

	\begin{description}

		\item[Case (i):] $g_m=g_{m+2}$. Then $\tilde{\Gamma}:=\big((g_1,...,g_m,\overset{\wedge}g_{m+1},\overset{\wedge}g_{m+2},g_{m+3},...,g_k),\tilde{\pi},c\big)\in \Pi^{-1}(\tilde{\pi})$ by definition (of $\pi-$adopted sequences). There are at most $N$ possibilities to colour $g_{m+1}$. That shows
		
		$$ \# \Pi^{-1}(\pi) \leq N\cdot \# \Pi^{-1}(\tilde{\pi}) \ . $$

		\item[Case (ii):] $g_m\neq g_{m+2}$. Then the same argument as in case (ii) in the proof of Proposition \ref{BlaetterEntfernungsProposition} gives 
		
		$$ \# \Pi^{-1}(\pi) \leq o(N^2) \cdot N^{k/2-1} \cdot B^{k/2-2} = o(N^{1+k/2}) \ . $$
	
	\end{description}

\end{proof}

\begin{lemma}\label{KreuzendePartitionenSindIrrelevant}
Let $\pi\in\mathcal{P}_2(\{ 1,...,k \})$ be a crossing pair-partition. Then one has

$$ \frac{1}{N^{k/2+1}} \cdot \# P^{-1}(\pi) \overset{N\rightarrow\infty}\longrightarrow 0 \ . $$

\end{lemma}

\begin{proof}
We first suppose that for all $m\in\{ 1,...,k \}$ the partition $\pi$ does not contain a block of the form $\{ m,m+1\}$.

\noindent Choose a \emph{minimal crossing} block, i.e.

$$ \{ m,m+l \} := \min_r \big\{ \{ m,m+r \}\in\pi \ | \ m\in\{ 1,...,k \}  \big\} \ . $$

\noindent We colour node $g_{m}$ arbitrarily, which gives $N$ possible values. Because of condition (\ref{AequivalenzrelationBedingung1}) and

$$ (c(g_m),c(g_{m+1})) \sim (c(g_{m+l}),c(g_{m+l+1})) \ ,  $$

\noindent there are most $o(N^2)$ possibilities to colour the nodes $g_{m+1},g_{m+l}$ and $g_{m+l+1}$. We colour any of the $l-2$ nodes in $\{ g_i | \ m+2\leq i \leq m+l-1 \}$ arbitrarily. That gives less than $N^{l-2}$ possibilities. Summarizing, all nodes in $\{ g_m,...,g_{m+l+1} \}$ are now coloured and we had less than 

 $$N\cdot o(N^2)\cdot N^{l-2}=o(N^{l+1})$$
 
\noindent possibilities to do so. To colour the other nodes successively, we consider the remaining $k/2-1$ blocks of $\pi$ ($\{m,m+l \}$ has already been used). Since $l$ was minimal in the above condition we find exactly $l-1$ blocks $B_1,...,B_{l-1}\in\pi$ with

$$ S:= \{ m+1,...,m+l-1 \} \ni m+s \in B_s =: \{ m+s, \beta_s \} \in\pi \textnormal{ and } \beta_s \not\in S \ .  $$

\noindent This means that all elements in $S$ \emph{cross} with elements outside $S$. The colour of $g_{\beta_s+1}$ is constrained by condition (\ref{AequivalenzrelationBedingung2}), because

$$  (c(g_{m+s}),c(g_{m+s+1})) \sim (c(g_{\beta_s}),c(g_{\beta_s+1})) \ .  $$

\noindent We further deduce that there are 

$$\frac{k}{2}-1-(l-1)=\frac{k}{2}-l$$

\noindent blocks $B_{l+1},...,B_{k/2}$ left which \emph{freely} colour $k/2-l$ nodes. This gives less than $N^{k/2-l}$ possibilities. The remaining 

$$k-(l+1)-\big(\frac{k}{2}-l\big)=\frac{k}{2}-1$$

\noindent nodes are then constrained by condition (\ref{AequivalenzrelationBedingung2}). This gives less than $B^{k/2-1}$ possible values. We conclude

$$ \# P^{-1}(\pi) \leq o(N^{l+1}) \cdot  N^{k/2-l} \cdot B^{k/2-1} = o(N^{k/2+1}) \ .  $$

\noindent If the partition $\pi$ consists of $s>0$ blocks of the form $\{ m,m+1 \}$ we use Proposition \ref{BlaetterEntfernungsProposition} to remove all these blocks. Then we gain an element $\tilde{\pi}\in\mathcal{P}_2(\{ 1,...,k-2s \})$ with $k-2s\geq 4$. There are no blocks of the form $\{ m,m+1 \}$ in $\tilde{\pi}$ and we have

$$ \# P^{-1}(\tilde{\pi}) \leq o(N^{(k-2s)/2+1}) $$

\noindent according to the case which has already been shown. The estimate in Proposition \ref{BlaetterEntfernungsProposition} gives

$$ \# P^{-1}(\pi) \leq N^s \cdot o(N^{(k-2s)/2+1}) + o(N^{k/2+1}) = o(N^{k/2+1}) \ . $$

\end{proof}

\begin{remark}
Lemma \ref{KreuzendePartitionenSindIrrelevant} shows that one can restict sum (\ref{MomentSummeUeberE}) to \emph{non-crossing} pair-partitions,

\begin{eqnarray}\label{SummeUberBaume}
\mu_k = \lim_{N\rightarrow\infty} \frac{1}{N^{1+k/2}} \sum_{\pi\in\mathcal{B}_{\frac{k}{2}}} \ \sum_{\Gamma\in\mathcal{E}^{(N)}_k(\pi)} \alpha(\Gamma) \cdot \mathbb{E}(X_N(\Gamma)) \ .
\end{eqnarray}
\end{remark}

We are ready to prove the main result of this section, that is how to calculate moments of the random matrix (\ref{Zufallsmatrix}). Prior we need to define what we mean by \emph{integration over trees}.

\begin{definition}\label{DefinitionIntgralUeberBaume}
For an even $k\in\mathbb{N}$ let $\pi\in\mathcal{B}_{\frac{k}{2}}$ be a non-crossing pair-partition and let $(\mathcal{G}_\pi;\mathcal{K})$ denote the $\pi-$adopted graph from Definition \ref{DefinitionAdaptierterGraph}  (\emph{rooted tree}). For a (bounded) Riemann-integrable function $\alpha: [0,1]\rightarrow\mathbb{R}$ we define

\begin{eqnarray}\label{DefinitionIntegral-J}
J_\alpha(\pi):= \underbrace{\int_0^1 \cdots \int_0^1}_{(\frac{k}{2}+1)-\textnormal{times}} \prod_{\{ i,j\}: \{ g_i,g_j \}\in\mathcal{K}} \alpha^2 (|x_i-x_j|) \ dx_1 \ \ldots \ dx_{\frac{k}{2}+1} \ .
\end{eqnarray}

\noindent $J_\alpha$ is called the $\alpha^2-$integral over $\pi$.

\end{definition}

\begin{theorem}\label{TheoremMomentGleichSummeIntegrale}
Let $\alpha:[0;1]\rightarrow\mathbb{R}$ be a bounded, Riemann-integrable function. Let

$$ M^{(N)} := \frac{1}{\sqrt{N}} \left( \alpha \big(\frac{|i-j|}{N} \big) \cdot X^{(N)}_{ij} \right)_{(1\leq i,j \leq N)}  $$

\noindent be a family of symmetric random matrices the entries $X^{(N)}_{ij}\equiv X^{(N)}_{ji}$ of which are centered and all moments exist. Let $\sim$ be an equivalence relation which satisfies conditions (\ref{AequivalenzrelationBedingung1})-(\ref{AequivalenzrelationBedingung3}) (and of course $(p,q)\sim (q,p)$ for all $p,q\in \{ 1,...,N \}$). We set

$$ \mu_k :=  \lim_{N\rightarrow\infty}  \frac{1}{N} \cdot\mathbb{E}\Big(\textnormal{tr} \big( (M^{(N)})^k\big)\Big) \ . $$

\noindent If $X^{(N)}_{pq}$ and $X^{(N)}_{rs}$ are assumed to be independent whenever $(p,q)\not\sim (r,s)$, then these moments can be computed by

\begin{eqnarray}\nonumber
\mu_k = \left\lbrace \begin{array}{cc}
\sum_{\pi\in\mathcal{B}_{\frac{k}{2}}} J_\alpha(\pi) & \textnormal{for $k$ even}
\\ 0 & \textnormal{otherwise.}
\end{array}\right. 
\end{eqnarray}

\end{theorem}

\begin{proof}
According to Proposition \ref{KorollarNurPaarPartitionenBeitrag} one has $\mu_{2k+1}=0$ for all $k\in\mathbb{N}$. Let $k$ from now on be even. In order to further analyze sum (\ref{SummeUberBaume}) we define

$$ \mathcal{E}^{(N)}_k \supset E_k := \big\{ \Gamma\in\mathcal{E}^{(N)}_k | (g_1,...,g_k) \textnormal{ is $\pi-$adopted} \big\} \ .   $$

\noindent We first show that

\begin{eqnarray}\label{NichtAdaptierteFolgenSindEgal}
\# \Pi^{-1}(\pi) \leq o(N^{k/2+1}) \ ,
\end{eqnarray}

\noindent where 

$$\Pi:= P|_{\overline{E_k}}:\overline{E_k}\longrightarrow\mathcal{B}_{\frac{k}{2}}$$

\noindent denotes the projector defined in Corollary \ref{BlaetterEntfernungsKorollarFurNichtadaptierteSequenzen}. Furthermore $\alpha$ is bounded on $[0;1]$. Therefore we again have
	
	$$ \alpha(\Gamma) \leq A=A(k) < \infty \ \qquad  \forall \ \Gamma\in\mathcal{E}_k \ . $$

\noindent This and \textsc{H\"older's} inequality imply that sum (\ref{SummeUberBaume}) can be restricted to $\Gamma\in E_k$ since

\begin{eqnarray}\nonumber
\frac{1}{N^{1+k/2}} \sum_{\Gamma\in\mathcal{E}^{(N)}_k(\pi)\backslash E_k} \alpha(\Gamma) \cdot \mathbb{E}(X_N(\Gamma))
&\leq& A\cdot R_k \cdot \frac{1}{N^{1+k/2}}\cdot \#\big(\mathcal{E}^{(N)}_k(\pi)\backslash E_k \big)
\\\nonumber &=&  A\cdot R_k \cdot \frac{1}{N^{1+k/2}} \cdot \# \Pi^{-1}(\pi)
\\\nonumber &\overset{N\rightarrow\infty}\longrightarrow& 0 \ . 
\end{eqnarray}

\noindent According to Lemma \ref{BaumeHabenBlatter} every $\pi\in\mathcal{B}_{\frac{k}{2}}$ has a block of the form $\{ m,m+1\}$. Therefore we can apply Corollary \ref{BlaetterEntfernungsKorollarFurNichtadaptierteSequenzen} as often as $k/2-1$ times. We arrive at $\tilde{\pi}=\{ \{ 1,2 \}\} \in\mathcal{B}_1$ and the estimate

$$ \# \Pi^{-1}(\pi) \leq N^{k/2-1} \cdot \# \Pi^{-1}(\tilde{\pi}) + o(N^{k/2+1}) $$

\noindent is valid. Since $(G_1,G_2)$ is not $\tilde{\pi}-$adopted, condition (\ref{AequivalenzrelationBedingung3}) gives the equality 

$$\# \Pi^{-1}(\tilde{\pi}) = o(N^2) \ .  $$

\noindent Therefore, inequality (\ref{NichtAdaptierteFolgenSindEgal}) is valid.

Next we use Lemma \ref{BlockegebenHinundZuruck} which states that

$$ X^{(N)}_{i_m i_{m+1}} \cdot X^{(N)}_{i_{m+l} i_{m+l+1}} = \big(X^{(N)}_{i_m i_{m+1}}\big)^2 $$

\noindent for all blocks $\{ m,m+l \}\in\pi$. Here we are using the notation $i_s=c(g_s)$ where $(g_1,...,g_k)$ denotes the unique $\pi-$adopted sequence. Because the random variables $X^{(N)}_{ij}$ have unit variance, it follows that

$$ \Gamma\in E_k \Longrightarrow \mathbb{E}(X_N(\Gamma)) = 1 \ , $$

\noindent since every random variable appears exactly twice. For $\Gamma_1,...,\Gamma_l\in E_k$ we set

$$ \alpha(\{ \Gamma_1,...,\Gamma_l \}):= \sum_{j=1}^l \alpha (\Gamma_j) \ . $$

\noindent  We now can rewrite sum (\ref{SummeUberBaume}) as

$$ \mu_k = \lim_{N\rightarrow\infty} \sum_{\pi\in\mathcal{B}_{\frac{k}{2}}} \ \frac{1}{N^{1+k/2}} \cdot \alpha(P|_E^{-1}(\pi)) \ , $$

\noindent where $P|_E$ denotes projector (\ref{ProjektorP}) restricted to the set $E_k$. It remains to show that

\begin{eqnarray}\label{IntegralGleichSumme}
J_\alpha(\pi) = \lim_{N\rightarrow\infty} \frac{1}{N^{1+k/2}} \cdot \alpha(P|_E^{-1}(\pi)) \ .
\end{eqnarray}

Now $\alpha$ is Riemann-integrable and therefore $\alpha^2$ is also Riemann-integrable since $[0;1]$ is compact. We can therefore approximate integral (\ref{DefinitionIntegral-J}) by Riemann-sums (using the equidistant partition of $[0;1]$). Using the notation

\begin{eqnarray}\nonumber
S_1&:=& \big\{ (r_1,...,r_{\frac{k}{2}+1})\in\{ 1,...,N \}^{k/2+1} | \ \exists\textnormal{ colour } c:\{ g_1,...,g_k\}\hookrightarrow \{ 1,...,N\} :  
\\\nonumber && r_s = c(g_s) \textnormal{ and $(g_1,...,g_k)$ is $\pi-$adopted} \big\} \textnormal{ and}
\\\nonumber S_2&:=& \{ 1,...,N \}^{k/2+1} \backslash S_1
\end{eqnarray}

\noindent we get

\begin{eqnarray}\nonumber
J_\alpha(\pi) &=& \int_0^1 \cdots \int_0^1 \prod_{\{ i,j\}: \{ g_i, g_j \}\in\mathcal{K}} \alpha^2 (|x_i-x_j|) \ dx_1 \ \ldots \ dx_{\frac{k}{2}+1}
\\\nonumber &=&  \lim_{N\rightarrow\infty} \frac{1}{N^{k/2+1}}\cdot \sum_{r_1,...,r_{\frac{k}{2}+1} =1}^N \underbrace{\prod_{\{ i,j\}: \{ g_i, g_j \}\in\mathcal{K}}  \alpha(\frac{|r_i-r_j|}{N}) \cdot \alpha(\frac{|r_j-r_i|}{N})}_{=:K(r)}
\\\nonumber &=& \lim_{N\rightarrow\infty} \frac{1}{N^{k/2+1}}\cdot \sum_{(r_1,...,r_{\frac{k}{2}+1}) \in S_1} K(r) + \lim_{N\rightarrow\infty} \frac{1}{N^{k/2+1}}\cdot \sum_{(r_1,...,r_{\frac{k}{2}+1}) \in S_2} K(r) 
\\\label{AbschaetzungTheorem} &=& \lim_{N\rightarrow\infty} \frac{1}{N^{k/2+1}}\cdot \sum_{\Gamma\in E_k} \alpha(\Gamma) + \lim_{N\rightarrow\infty} \frac{1}{N^{k/2+1}}\cdot \sum_{\Gamma\in \mathcal{E}_k\backslash E_k} \alpha(\Gamma) 
\\\nonumber &\geq& \lim_{N\rightarrow\infty} \frac{1}{N^{k/2+1}}\cdot \sum_{\Gamma\in E_k} \alpha(\Gamma) 
\\\nonumber &=& \lim_{N\rightarrow\infty} \frac{1}{N^{1+k/2}} \cdot \alpha(P|_E^{-1}(\pi)) 
\end{eqnarray}

\noindent and it remains to show '$\leq$' in (\ref{IntegralGleichSumme}). Using again the above notation, inequation (\ref{NichtAdaptierteFolgenSindEgal}) implies

\begin{eqnarray}\nonumber
(\ref{AbschaetzungTheorem}) &\leq& \lim_{N\rightarrow\infty} \frac{1}{N^{k/2+1}}\cdot \sum_{\Gamma\in E_k} \alpha(\Gamma) + \lim_{N\rightarrow\infty} \frac{1}{N^{k/2+1}} \cdot A \cdot \# (\mathcal{E}_k(\pi)\backslash E_k(\pi))
\\\nonumber &\leq& \lim_{N\rightarrow\infty} \frac{1}{N^{k/2+1}}\cdot \sum_{\Gamma\in E_k} \alpha(\Gamma)
\\\nonumber &=& \lim_{N\rightarrow\infty} \frac{1}{N^{k/2+1}}\cdot \alpha(P|_E^{-1}(\pi)) \ . 
\end{eqnarray}

\end{proof}

\begin{remark} Theorem \ref{TheoremMomentGleichSummeIntegrale} generalizes Schenker's and Schulz-Baldes' result in \cite{HSBS} since they consider the case $\alpha\equiv 1$. In this case 

$$ J_\alpha(\pi) \equiv 1 $$

\noindent is valid and we gain $\mu_k= \# \mathcal{B}_{\frac{k}{2}}$ for even $k$. The validity of the Semicircular Law is gained by the following
\end{remark}

\begin{lemma}\label{CatalanZahl}
Let $C_k$ denote the $k-$th \emph{Catalan-}number, which is

$$ C_k := \frac{1}{k+1} \cdot \left( \begin{array}{c}
2k \\ k
\end{array}\right) \ . $$

\noindent Then $\# \mathcal{B}_{k}=C_{k}$.

\end{lemma}

\begin{proof}
See e.g. \cite{AGZ}, \cite{HSBS}.
\end{proof}

\newpage

\section{Applications: The Semicircular Law}\label{section:Applications}

In this section we discuss conditions which are necessary and sufficient for the validity of the Semicircular Law for random matrices $M^{(N)}$ defined in Theorem \ref{TheoremMomentGleichSummeIntegrale}. During the whole section, for simplicity, we say that the SCL is valid, if 

\begin{eqnarray}\nonumber
\mu_k^{(N)} \overset{N\rightarrow\infty}\longrightarrow 
\left\lbrace \begin{array}{cc}
C_{\frac{k}{2}} & \textnormal{for $k$ even}
\\ 0 & \textnormal{otherwise.}
\end{array}\right.
\end{eqnarray}

\noindent Weak convergence in probability, as explained in (\ref{IntroHalbkreisgesetzFormulierung}), will be shown in Theorem \ref{TheoremSchwacheKonvergenzinWkeitVollerFall}. Furthermore, we consider families of random matrices the entries of which are correlated in the sense of an equivalence-relation satisfying conditions (\ref{AequivalenzrelationBedingung1})-(\ref{AequivalenzrelationBedingung3}).

Let $A$ and $B$ denote random matrices the entries of which are independent, centralized and normalized. As an application of the results obtained in the first part of this section we will discuss random \emph{block-matrices}, i.e.

\begin{eqnarray}\nonumber
M^{(N)}:= \frac{1}{\sqrt{2N}} \cdot \left( \begin{array}{cc}
A & B \\
B^T & \pm A
\end{array} \right).
\end{eqnarray}

\noindent These (random) block-matrices turn out to be included in the developed theory. Finally, this section is concluded with the proof for the validity of the SCL (in the above sense) for random \emph{band-matrices} the band-width of which grows as $o(N)$. This case will be referred to as \emph{slow-growing} band-width.

\subsection{Integral kernels and the Semicircular Law}

\begin{lemma}\label{Reduktionslemma}
We assume all requirements from Definition \ref{DefinitionIntgralUeberBaume}. Furthermore we define

\begin{eqnarray}\nonumber
\varphi(x) &:=&\int_0^1 \alpha^2(|x-y|) \ dy \hspace{3ex}\textnormal{ and}
\\\nonumber J_\alpha(\pi|x_s) &:=& \underbrace{\int_0^1 \cdots \int_0^1}_{\frac{k}{2}-\textnormal{times}} \prod_{\{ i,j\}: \{ g_i, g_j \}\in\mathcal{K}} \alpha^2(|x_i-x_j|) \ dx_1 \ldots \ \overset{\wedge}{dx_s} \ \ldots \ dx_{\frac{k}{2}+1} \ .
\end{eqnarray}

\noindent Then for all blocks $\{ m,m+1 \}\in\pi$ the recursion

$$ J_\alpha(\pi) =   \int_0^1 \varphi(x_{m})\cdot J_\alpha(\tilde{\pi}|x_{m}) \ dx_{m}   $$

\noindent with $\tilde{\pi}=\pi\setminus^\bullet \{ m,m+1 \}$ is valid.
\end{lemma}

\begin{proof}
According to Fubini we have

$$ J_\alpha(\pi)= \int_0^1 J_\alpha(\pi|x_s) \ dx_s \ \ \qquad \forall \ s\in \{ 1,..., \frac{k}{2}+1 \} \ . $$

\noindent By Lemma \ref{BaumeHabenBlatter} there exists at least one block $\{ m,m+1 \} \in\pi$. For $\tilde{\pi}=\pi\setminus^\bullet \{ m,m+1 \}$ let $(\mathcal{G}_{\tilde{\pi}}, \mathcal{K}_{\tilde{\pi}})$ denote the $\tilde{\pi}-$adopted graph from Definition \ref{DefinitionAdaptierterGraph}. Because $m+1$ is a \emph{leaf} which is only connected with $m$ (in the sense of the mentioned definition) we have

\begin{eqnarray}\nonumber
J_\alpha(\pi|x_m) &=& \underbrace{\int_0^1 \int_0^1 \ldots \int_0^1}_{(\frac{k}{2}-1)-\textnormal{times}} \ \left(  \int_0^1 \alpha^2(|x_m-x_{m+1}|) \ dx_{m+1} \right) \times 
\\\nonumber && \prod_{ \{ r,s\}: \{ g_r, g_s \}\in\mathcal{K}_{\tilde{\pi}}} \alpha^2(|x_r-x_s|) \ dx_1 \ \ldots \  \overset{\wedge}{dx_m} \ \overset{\wedge}{dx_{m+1}} \ \ldots \ dx_{\frac{k}{2}+1}
\\\nonumber &=&  \underbrace{\int_0^1 \int_0^1 \ldots \int_0^1}_{(\frac{k}{2}-1)-\textnormal{times}} \prod_{ \{ r,s\}: \{ g_r, g_s \}\in\mathcal{K}_{\tilde{\pi}} }\alpha^2(|x_r-x_s|) \ d(x_1, x_2,... \overset{\wedge}{x_m}, \overset{\wedge}{x_{m+1}},...,  x_{\frac{k}{2}+1})\cdot \varphi(x_m)
\\\nonumber &=& J(\tilde{\pi}|x_m)\cdot \varphi(x_m) \ ,
\end{eqnarray}

\noindent which completes the proof.
\end{proof}

\begin{remark}
Suppose that $\varphi(x)\equiv\varphi_0$. Then we obtain 

$$ J_\alpha(\pi) = \varphi_0^{k/2} \ \qquad \forall \ \pi\in\mathcal{B}_{\frac{k}{2}}   $$

\noindent by Lemma \ref{Reduktionslemma} and an induction. This will be a key observation in proving the SCL, see below.

\end{remark}

\begin{example}\label{BspPeriodischeBandMatrizen}
Let $f:[0;N]\rightarrow\mathbb{R}$ be a function and let 

$$A_N:=\frac{1}{\sqrt{N}}\cdot \left(f(|i-j|)\cdot X^{(N)}_{ij}\right) _{(1\leq i,j \leq N)}$$

\noindent be an ensemble of random matrices. The random variables $X^{(N)}_{ij}$ may be centered and have unit variance. Furthermore they may be correlated in the sense of an equivalence-relation satisfying conditions (\ref{AequivalenzrelationBedingung1})-(\ref{AequivalenzrelationBedingung3}). On $\{ 1,...,N \}$ we consider

$$ |i|_N := \min \{ i; N-i \} \ .  $$

\noindent The ensemble $A_N$ is called \emph{periodic} with bandwith $b_N$, if


$$ f = \chi_{[0;b_N]}(|\cdot|_N) \ .  $$

\noindent One can ask the question weather the SCL holds for a proportional growth, that is $b_N=\rho N$, where $0<\rho<\frac{1}{2}$ is a real number. This example was also discussed in \cite{BMP} for the case of independent random variables.

If we set $\alpha:= \chi_{[0;\rho]}+\chi_{[1-\rho;1]}$, then $\alpha$ is clearly bounded and Riemann-integrable. We have
$$ \varphi(x)=\int_0^1 \alpha^2(|x-y|) \ dy = \int_0^1 \alpha(|x-y|) \ dy = \rho+ \big( 1-(1-\rho) \big)\equiv 2\rho \ . $$

\noindent This implies 

$$J_{\alpha}(\pi)=(2\rho)^{k/2} $$

\noindent as the remark after Lemma \ref{Reduktionslemma} shows. Therefore the $k:=2m-$th moment of 

$$ \tilde{M}^{(N)} := \frac{1}{\sqrt{2\rho}} \cdot M^{(N)} \ , $$

\noindent where $M^{(N)}$ is the family of random matrices defined in Theorem \ref{TheoremMomentGleichSummeIntegrale}, is computed by

\begin{eqnarray}\nonumber
\mu_{2m}=\frac{1}{(2\rho)^m} \cdot \sum_{\pi\in\mathcal{B}_m} J_\alpha(\pi) =\frac{1}{(2\rho)^m} \cdot \sum_{\pi\in\mathcal{B}_m}(2\rho)^{m} =  {C_m}  \ ,
\end{eqnarray}

\noindent while the odd moments are zero. This implies the validity of the SCL for $\tilde{M}^{(N)}$. These matrices are indeed the periodic band-matrices with bandwith $b_N= \rho N$ since $\alpha(x)\in\{ 0,1 \}$ and 

\begin{eqnarray}\nonumber
\alpha(\frac{|i-j|}{N}) = 1 &\Longleftrightarrow & |i-j|\leq\rho N \textnormal{ or }  (1-\rho)N\leq |i-j| \leq N
\\\nonumber &\Longleftrightarrow&  |i-j|\leq\rho N \textnormal{ or }  N-|i-j| \leq \rho N
\\\nonumber &\Longleftrightarrow& \big||i-j|\big|_N= \min \{ |i-j|; N-|i-j| \} \leq \rho N \ .
\end{eqnarray}	

\noindent This example shows, that 

$$ \varphi(x) \equiv \varphi_0  $$

\noindent is sufficient for the validity of the SCL. In particular we get the SCL for the Wigner-case by starting with $\rho=\frac{1}{2}$. Furthermore, this example shows that we can take the limit $\rho\rightarrow\frac{1}{2}$ and \emph{also} get convergence (of the moments) to the \textsc{Catalan}-numbers.
\end{example}

\begin{remark}
The result from Example \ref{BspPeriodischeBandMatrizen} is well known, see \cite{BMP}.
\end{remark}

\noindent We will show next, that the coniditon $\varphi(x) \equiv \varphi_0$ is also necessary for the validity of the SCL.

\begin{theorem}\label{TheoremNotwendigeUndHinreichendeBedfuerSCL}
We consider a symmetric ensemble $M^{(N)}:=\frac{1}{\sqrt{N}} \left( \alpha(\frac{|i-j|}{N}) \cdot X^{(N)}_{ij} \right)_{(1\leq i,j \leq N)}$ with all requirements from Theorem \ref{TheoremMomentGleichSummeIntegrale}. Furthermore we define


$$ \varphi(x)=\int_0^1 \alpha^2(|x-y|) \ dy \textnormal{ and } \varphi_0 := \int_0^1 \varphi(x) \ dx \ . $$

\noindent Then the SCL for $\frac{1}{\sqrt{\varphi_0}}\cdot M^{(N)}$ holds if -- and only if -- 

$$ \varphi(x) \equiv \varphi_0 \ . $$

\noindent Furthermore, the moments $\mu_k$ of $M^{(N)}$ can in principle be computed by the formula given in Theorem \ref{TheoremMomentGleichSummeIntegrale}.

\end{theorem}

\begin{proof} We first show

\underline{'$\Longrightarrow$':} The remark after Lemma \ref{Reduktionslemma} shows that

$$ J_\alpha(\pi) \equiv \varphi_0^{k/2} \ .  $$

\noindent Theorem \ref{TheoremMomentGleichSummeIntegrale} implies

\begin{eqnarray}\nonumber
\mu_k = \frac{1}{\varphi_0^{k/2}}\cdot\sum_{\pi\in\mathcal{B}_{\frac{k}{2}}} \varphi_0^{k/2} = C_{\frac{k}{2}} 
\end{eqnarray}

\noindent for even $k\in\mathbb{N}$ while the odd moments are zero.

\smallskip
\underline{'$\Longleftarrow$':} Suppose that $\varphi$ is not constant on $[0;1]$. We show that the SCL in this case does not hold. In order to do so, we remark, that $\varphi\in L_2([0;1])$ since $\alpha$ is Riemann-integrable. We will use Cauchy-Schwartz-inequation (CSI) in order to show

$$ C_2 < \mu_4 \ . $$

Since $\mathcal{B}_2 = \{ \pi_1:=\{ \{1,2\} , \{ 3,4 \}\}, \pi_2:=\{ \{1,4\} , \{ 2,3 \}\} \}$, the according adopted sequences are 

$$ G^{(1)}=(g_1,g_2,g_1,g_3)  \textnormal{ and }  G^{(2)}=(g_1,g_2,g_3,g_2) \ , \ g_i\neq g_j \textnormal{ for } i\neq j \ , $$

\noindent as can easily be checked. Therefore the edges are

\begin{eqnarray}\nonumber
\mathcal{K}_{\pi_1} &=& \{ \{ g_1,g_2\}, \{ g_1,g_3\} \} =: \{ \{ x,y \} , \{ x,z\} \} \textnormal{ and }
\\\nonumber \mathcal{K}_{\pi_2} &=& \{ \{ g_1,g_2\}, \{ g_2,g_3\} \} =: \{ \{ x,y \} , \{ y,z\} \} \ . 
\end{eqnarray}

\noindent According to \textsc{Fubini} we have

$$ \varphi(x)^2 = \int_0^1 \int_0^1 \alpha^2(|x-y|)\cdot \alpha^2(|x-z|) \ d(y,z) \ .  $$

\noindent Therefore $J_\alpha(\pi_{1,2})$ can be calculated as

\begin{eqnarray}\nonumber
J_\alpha(\pi_1) &=& \int_0^1 \varphi(x)^2 \ dx \ \textnormal{ and }\ J(\pi_2) = \int_0^1 \varphi(y)^2 \ dy 
\\\nonumber \Longrightarrow \mu_4 &=& \frac{2}{\varphi_0^2}\cdot\int_0^1 \varphi(x)^2 \ dx \ . 
\end{eqnarray}

\noindent Since $\mathcal{B}_1= \{ \{ 1,2\} \}$, one gets 

$$ \mu_2 = \frac{1}{\varphi_0} \cdot \int_0^1\int_0^1 \alpha^2(|x-y|) \ dy \ dx = \frac{1}{\varphi_0} \cdot \int_0^1 \varphi(x) \ dx = 1 \ . $$

\noindent Therefore

\begin{eqnarray}\nonumber
C_2= 2\cdot\mu_2^2 &=&2\cdot \frac{1}{\varphi_0^2}\cdot \big|\int_0^1 \varphi(x) \ dx \big|^2
\\\nonumber &=& 2\cdot \frac{1}{\varphi_0^2}\cdot \big|<\varphi, \chi_{[0;1]}>\big|^2
\\\nonumber &\underset{CSI}<&  \frac{2}{\varphi_0^2}\cdot || \varphi||^2 \cdot \underbrace{|| \chi_{[0;1]}||^2}_{=1}
\\\nonumber &=& \frac{2}{\varphi_0^2}\cdot \int_0^1 \varphi(x)^2 \ dx
\\\nonumber &=& \mu_4 \ ,
\end{eqnarray}

\noindent since $\varphi$ is not constant on $[0;1]$ and therefore it is linear independent with $\chi_{[0;1]}$.
\end{proof}

\begin{corollary}\label{KorollarPeriodischeBandMatrizen}
For a real number $\rho\in [0;1]$ we consider the ensemble $ M^{(N)}$ defined in Theorem \ref{TheoremMomentGleichSummeIntegrale}. We set

$$ \alpha := \chi_{[0;\rho]} \ .  $$

\noindent These matrices are called (non-periodic) \emph{band-matrices} with proportional growth, since

$$ \alpha\big(\frac{|i-j|}{N}\big) = 1 \Longleftrightarrow |i-j| \leq \rho N =: b_N \ .  $$

\noindent The SCL does not hold in this case.
\end{corollary}

\begin{remark}
This result is well-known for independent random variables and $\alpha=\textnormal{const.}$, see \cite{BMP}.
\end{remark}

\begin{proof}
According to Theorem \ref{TheoremNotwendigeUndHinreichendeBedfuerSCL} it is sufficient to show, that $\varphi \neq \varphi_0$. First consider $0<\rho < \frac{1}{2}$. Then we have 

 $$ \alpha^2(|x-y|)=1 \Longleftrightarrow |x-y|\leq\rho \Longleftrightarrow x-\rho \leq y \leq x+\rho \ . $$

\noindent The zero of $y=x-\rho$ is $x=\rho$. Furthermore the intersection of the function $y=x+\rho$ with $y=1$ is $x=1-\rho$. Because of

 $$0<\rho < \frac{1}{2} \Longleftrightarrow \rho<1-\rho$$

\noindent we have to consider the following three cases:


	
	
	
	\begin{description}
	
		\item[(i)] $0< x < \rho$: $\varphi(x)=\int_0^{{x+\rho}}1 \ dy = x+\rho$.

		\item[(ii)] $\rho\leq x < 1-\rho$: $\varphi(x)=\int_{x-\rho}^{{x+\rho}}1 \ dy = 2\rho$.
		
		\item[(iii)] $1-\rho \leq x < 1$: $\varphi(x)=\int_{x-\rho}^1 1 \ dy = 1+\rho-x$.
		
	\end{description}
	

\noindent The case $\frac{1}{2}\leq \rho <1$ is done analogously:

\begin{eqnarray}\nonumber
\varphi(x)=\left\lbrace  \begin{array}{cl}
x+\rho & \textnormal{ for } 0<x< 1-\rho \\ 
1 & \textnormal{ for } 1-\rho<x<\rho \\
1+\rho-x & \textnormal{ for } \rho<x<1
\end{array} \right. 
\end{eqnarray}

\noindent Therefore $\varphi(x)$ is not constant and the SCL does not hold. 

\end{proof}

\begin{remark}
If we take the limit $\rho\rightarrow 1$, then Corollary \ref{KorollarPeriodischeBandMatrizen} results in $\varphi(x)\equiv 1 = \varphi_0$. This shows that - again - the moments of $M^{(N)}$ converge against those of the SCL for $\rho\rightarrow 1$. On the other hand, starting with $\rho=1$, one gets the Wigner-case.
\end{remark}

\subsection{Random Block-Matrices}

In this section we consider (random) block-matrices, i.e.

\begin{eqnarray}\nonumber
M^{(N)}:= \frac{1}{\sqrt{2N}} \cdot \left( \begin{array}{cc}
A & B \\
B^T & \pm A
\end{array} \right) =: \frac{1}{\sqrt{2N}}\cdot\big( \xi_{ij}^{(N)} \big)_{1\leq i,j \leq 2N} \ . 
\end{eqnarray}

\begin{theorem}\label{TheoremBlockMatrizen}
Let $A$ be a symmetric \textsc{Wigner}-type matrix the entries of which are centralized and normalized. Let $B$ denote a not necessarily symmetric \textsc{Wigner}-type matrix with centralized and normalized entries. Furthermore, $A$ and $B$ are considered to be independent, i.e. every random variable $X$ taken from $A$ and every $Y$ taken from $B$ are independent. Consider the symmetric ensemble

\begin{eqnarray}\nonumber
\alpha\cdot M^{(N)} := \frac{1}{\sqrt{2N}}\cdot\big( \alpha\big( \frac{|i-j|}{2N} \big) \cdot \xi_{ij}^{(N)} \big)_{1\leq i,j \leq 2N} \ ,
\end{eqnarray}

\noindent where $\alpha: [0,1]\rightarrow\mathbb{R}$ denotes a bounded, Riemann-integrable function. Let $\mu_k^{(N)}$ denote the $k-$th moment of $\alpha\cdot M^{(N)}$. Then we have 

\begin{eqnarray}\nonumber
\lim_{N\rightarrow\infty} \mu_k^{(N)} = \left\lbrace \begin{array}{cc}
\sum_{\pi\in\mathcal{B}_{\frac{k}{2}}} J_\alpha(\pi) & \textnormal{for $k$ even}
\\ 0 & \textnormal{otherwise.}
\end{array}\right.
\end{eqnarray}
\end{theorem}

\smallskip

\begin{remark}
Ensembles of this form with the special weight $\alpha\equiv 1$ were discussed in \cite{HCS}. On the other hand, Theorem \ref{TheoremBlockMatrizen} shows that all results obtained in the previous section hold also for the mentioned block-matrices. In particular, the SCL does not hold for random band-block-matrices the band-width of which is proportional to its dimension, see Corollary \ref{KorollarPeriodischeBandMatrizen}. On the other hand, the SCL is valid for periodic band-block-matrices, see Example \ref{BspPeriodischeBandMatrizen}.
\end{remark}

\begin{proof}(of Theorem \ref{TheoremBlockMatrizen}.)
Consider the ensemble

$$ \frac{1}{\sqrt{2N}}\cdot\big( \alpha\big( \frac{|i-j|}{2N} \big) \cdot \xi_{ij}^{(N)} \big)_{1\leq i,j \leq 2N} \ . $$

\noindent This ensemble is considered to be a \emph{single random matrix} the entries of which are correlated in a certain way. For simplicity, $A$, $\pm A$, $B$ and $B^T$ are called \emph{blocks} of $\alpha\cdot M^{(N)}$. In order to proof the Theorem it is sufficient to verify that conditions (\ref{AequivalenzrelationBedingung1})-(\ref{AequivalenzrelationBedingung3}) are fullfilled. Theorem \ref{TheoremMomentGleichSummeIntegrale} then results in Theorem \ref{TheoremBlockMatrizen}.

\begin{description}

	\item[Condition (\ref{AequivalenzrelationBedingung1}):] Let $p\in\{ 1,...,2N\}$ be fixed. One has then $2N$ possibilities to choose $q$. This defines a random variable $\xi^{(N)}_{pq}$ in one of the blocks. Since $A$ and $B$ are assumed to be independent, there are at most 3 random variables which are correlated to $\xi^{(N)}_{pq}$ as can easily be seen: \emph{the same} $\xi^{(N)}_{pq}$ is found one time on the other side of the diagonal of $\alpha\cdot M^{(N)}$. At most two of the same random variables are found on the other side of the diagonal \emph{inside} the mentioned blocks. This gives most $6N$ choices.

	\item[Condition (\ref{AequivalenzrelationBedingung2}):] Let $(p,q,r)\in\{ 1,...,2N\}^3$ be fixed. This gives a random variable $\xi_{pq}^{(N)}$ in one of the blocks. Now we look at row $r$ in $\alpha\cdot M^{(N)}$. Then there is at most one other random variable correlated to $\xi_{pq}^{(N)}$ since, again, $A$ and $B$ are assumed to be independent.

	\item[Condition (\ref{AequivalenzrelationBedingung3}):] We consider two pairs $(p,q)\sim(q,r)$ and ignore $p\neq r$ for an upper bound. If the random variable $\xi_{pq}^{(N)}$ belongs to $A$ or $B^T$, then row $q$ is found where blocks $A$ and $B$ are. Otherwise row $q$ is found where blocks $B^T$ and $\pm A$ are. Therefore, the total number of correlated random variables in row $q$ is at most one, since $A$ and $B$ are assumed to be independent and they are both of the \textsc{Wigner-}type.
	
\end{description}

\end{proof}

\subsection{Band-Matrices with slow-growing Band-Width}

In this section we consider band-matrices the band-width of which behaves like $o(N)$. It is clear, that conditions (\ref{AequivalenzrelationBedingung1})-(\ref{AequivalenzrelationBedingung3}) have to be modified in order to get the validity of the SCL.

\begin{theorem}\label{TheoremKleinOvonN}
Set $b_N=o(N)$ and suppose $b_N\overset{N\rightarrow\infty}\longrightarrow\infty$. Let $\sim$ be an equivalence relation on $\{ 1,...,N \}^2$ satisfying the following conditions:

\begin{eqnarray}\label{BedingungO1}
\max_{p} &\#& \{ (q,r,s)\in\{ 1,...,N\}^3 | (p,q)\sim (r,s) \} = o(b_N^2)
\\ \label{BedingungO2} \max_{p,q,r} &\#& \{ s\in\{ 1,...,N \} | (p,q)\sim(r,s) \} \leq B < +\infty
\\ \label{BedingungO3}  &\#& \{ (p,q,r)\in\{ 1,...,N \}^3 | (p,q)\sim(q,r) \textnormal{ and } r\neq p \} = o(b_N^2)
\\\nonumber && (p,q)\sim (q,p) \ \forall \ p,q \in\{ 1, ..., N \} \ .
\end{eqnarray}

\noindent Consider the symmetric ensemble

$$ A^{(N)}:=\frac{1}{\sqrt{2b_N}} \left( \chi_{[0;b_N]} \big(|i-j| \big) \cdot X^{(N)}_{ij} \right)_{(1\leq i,j \leq N)} $$

\noindent with centered random matrices $X^{(N)}_{ij}\equiv X^{(N)}_{ji}$. All (centralized) moments of $X^{(N)}$ may exist and may be bounded. If the random variables $X^{(N)}_{pq}$ and $X^{(N)}_{rs}$ are independent, whenever $(p,q)\not\sim (r,s)$, then the SCL holds for $A^{(N)}$.

\end{theorem}

\begin{proof}
We will modify Lemma \ref{KreuzendePartitionenSindIrrelevant} in order to show, that only non-crossing pair-partitions have an asymptotic contribution to the $k-$th moment

\begin{eqnarray}\nonumber
\mu_k &=& \lim_{N\rightarrow\infty}\frac{1}{N} \cdot \mathbb{E}(\textnormal{tr}(A^{(N)})^k)
\\\label{O-N-Summe} &=& \lim_{N\rightarrow\infty}\frac{1}{N\cdot (2b_N)^{k/2}} \cdot \sum_{\pi\in\mathcal{P}(k)} \ \sum_{\Gamma\in\mathcal{F}^{(N)}_k(\pi)} \mathbb{E}(X_N(\Gamma)) \ \ \textnormal{ with}
\\\nonumber \mathcal{F}^{(N)}_k(\pi) &=& \big\{ \Gamma=\big((g_1,...,g_k),\pi,c\big)\in \mathcal{E}^{(N)}_k(\pi) | \ |c(g_m)-c(g_{m+1})| \leq b_N \ \forall \ 1\leq m \leq k  \big\} \ . 
\end{eqnarray}

In order to do so, we first agree, that only \emph{pair-partitions} give an asymptotically non-vanishing contribution to (\ref{O-N-Summe}). Therefore we modify the above argumentation. Let $\pi=\{ B_1, ..., B_r\} \in\mathcal{P}(k)$ be a partition and consider the projector 

\begin{eqnarray}\nonumber
P: \mathcal{F}^{(N)}_k(\pi) &\longrightarrow& \mathcal{P}(k)
\\\nonumber \Gamma &\longmapsto& \pi \ .
\end{eqnarray}

\noindent In order to count $P^{-1}(\pi)$ we distinguish between two cases.

\begin{description}

	\item[Case (i):] $c(g_1)\in \{ k\cdot b_N+1, ..., N-k\cdot b_N \}$. According to (\ref{O-N-Summe}) we have $2b_N$ possibilities to colour $g_2$. The other nodes will be coloured successively. This is done analogously to the proof of Lemma \ref{ObereAbschatzung}. Since we always have less than $2b_N$ possibilities to colour each node, we gain
	
	$$ \# P^{-1}(\pi) \leq (N-2k\cdot b_N) \cdot (2b_N) \cdot (2b_N)^{r-1} \cdot B^{k-r-1} \ . $$

	\item[Case (ii):] $c(g_1)\not\in \{ k\cdot b_N+1, ..., N-k\cdot b_N \}$. Analogously to the case above we have less than $2b_N$ possibilities to colour each node. Therefore we have
	
	$$ \# P^{-1}(\pi) \leq (2k\cdot b_N) \cdot (2b_N) \cdot (2b_N)^{r-1} \cdot B^{k-r-1} = k \cdot (2b_N)^{r+1}\cdot B^{k-r-1}   \ . $$
\end{description}

\noindent Combining the two cases, we get

$$ \# P^{-1}(\pi) \leq N \cdot (2b_N)^{r} \cdot B^{k-r-1} \ .  $$

\noindent For $r<\frac{k}{2}$ we deduce

$$ \frac{1}{N\cdot (2b_N)^{k/2}} \cdot \# \mathcal{F}^{(N)}_k(\pi) \leq (2b_N)^{r-k/2} \cdot B^{k-r-1} \overset{N\rightarrow\infty}\longrightarrow 0 \ . $$

\noindent The case $r>\frac{k}{2}$ gives a singleton block as in Proposition \ref{KorollarNurPaarPartitionenBeitrag}. This means that there is a random variable which appears exactly once. Since its expectation is zero, the corresponding summand is zero.

Next we suppose that $\pi$ is a \emph{crossing} pair-partition with no blocks of the form $\{ m,m+1 \}$. Analogously to Lemma \ref{KreuzendePartitionenSindIrrelevant} we choose a block $\{ m, m+l\}\in\pi$ with minimal $l\geq 2$ and colour node $g_m$.

\begin{description}

	\item[Case (i):] $c(g_m)\in \{ k\cdot b_N+1, ..., N-k\cdot b_N \}$. Using the notation $c_s:=c(g_s)$ we gain
	
	$$ \# \{ (c_{m+1},c_{m+l},c_{m+l+1}) | (c_m,c_{m+1}) \sim (c_{m+l},c_{m+l+1}) \} = o(b_N^2) $$
	
	\noindent according to Condition (\ref{BedingungO1}). Furthermore we can \emph{freely} colour all nodes
	
	$$ g \in  \{ g_{m+2}, ..., g_{m+l-1} \} $$
	
	\noindent which gives less than $(2b_N)^{l-2}$ possibilities. After analyzing the proof of Lemma \ref{KreuzendePartitionenSindIrrelevant} this shows that
	
	\begin{eqnarray}\nonumber
	\# P^{-1}(\pi) &\leq&  (N-2k\cdot b_N) \cdot (2b_N)^{l-2} \cdot o(b_N^2) \cdot (2b_N)^{k/2-l} \cdot B^{k/2-1}
	\\\nonumber &=& (N-2k\cdot b_N) \cdot o(b_N^{k/2}) \ .
	\end{eqnarray}
	
	\item[Case (ii):] $c(g_m)\not\in \{ k\cdot b_N+1, ..., N-k\cdot b_N \}$. The only new feature is, that we only have $2k\cdot b_N$ possibilities to colour $g_m$. That analogously shows that
	
	$$ \# P^{-1}(\pi) \leq  (2k\cdot b_N) \cdot (2b_N)^{l-2} \cdot o(b_N^2) \cdot (2b_N)^{k/2-l} \cdot B^{k/2-1} = o(b_N^{k/2+1}) \ . $$
	
\end{description}

\noindent Since $b_N=o(N)$, both inequalities show

$$ \frac{1}{N\cdot (2b_N)^{k/2}}\cdot \# P^{-1}(\pi) \overset{N\rightarrow\infty}\longrightarrow 0 \ . $$

\noindent For the case that there are blocks of the form $\{ m,m+1 \}$, one can remove these blocks using an inequality of the form

\begin{eqnarray}\label{ObereAbschaetzung-O-N}
\# P^{-1}(\pi) \leq (2b_N)\cdot \# P^{-1}(\tilde{\pi}) + o(b_N^{k/2+1}) \ .
\end{eqnarray}

\noindent This inequality is gained from the proof of Proposition \ref{BlaetterEntfernungsProposition}. According to (\ref{O-N-Summe}), the $k:=2m-$th moment of $A^{(N)}$ can therefore be computed by

\begin{eqnarray}\label{Momente-O-N-Summe}
\mu_{2m} = \lim_{N\rightarrow\infty}\frac{1}{N\cdot (2b_N)^{m}} \cdot \sum_{\pi\in\mathcal{B}_{m}} \ \sum_{\Gamma\in\mathcal{F}^{(N)}_{2m}(\pi)} \mathbb{E}(X_N(\Gamma)) \ , 
\end{eqnarray}

\noindent while the odd moments vanish. 

In order to show that only $\pi-$adopted sequences give a contribution to sum (\ref{Momente-O-N-Summe}) we define 

$$ \mathcal{F}^{(N)}_k \supset F_k := \big\{ \Gamma\in\mathcal{F}^{(N)}_k | (g_1,...,g_k) \textnormal{ is $\pi-$adopted} \big\} \ .   $$

\noindent The proof of Corollary \ref{BlaetterEntfernungsKorollarFurNichtadaptierteSequenzen} shows that in our case we \emph{also} have an inequality of the form

$$ \# \Pi^{-1}(\pi) \leq (2b_N) \cdot \# \Pi^{-1}(\tilde{\pi}) + o(b_N^{k/2+1}) \ .  $$

\noindent Here we are using the notations $\tilde{\pi}:=\pi\backslash^\bullet\{ m,m+1 \}\in\mathcal{B}_{\frac{k-2}{2}}$ and

$$ \Pi:= P|_{\overline{F_k}}:\overline{F_k}\longrightarrow\mathcal{B}_{\frac{k}{2}} \textnormal{ with } \overline{F_k}:= \mathcal{F}^{(N)}_k \backslash F_k \ . $$

\noindent The only new feature in this case is, that by definition of $\mathcal{F}^{(N)}_k$, one only has $2b_N$ possibilities to colour a node. Now Lemma \ref{BaumeHabenBlatter} can be applied as often as $k/2-1$ times to gain an element $\tilde{\pi}=\{ \{ 1,2 \}\} \in\mathcal{B}_1$. We get the estimate

$$ \# \Pi^{-1}(\pi) \leq (2b_N)^{k/2-1} \cdot \# \Pi^{-1}(\tilde{\pi}) + o(b_N^{k/2+1}) \ . $$

\noindent Since $(G_1,G_2)$ is not $\tilde{\pi}-$adopted, condition (\ref{BedingungO3}) gives the equality 

$$\# \Pi^{-1}(\tilde{\pi}) = o(b_N^2) \ .  $$

\noindent This equality implies

\begin{eqnarray}\nonumber
\frac{1}{N\cdot (2b_N)^{k/2}} \cdot \# \Pi^{-1}(\pi) &\leq& \frac{1}{N\cdot (2b_N)^{k/2}} \cdot \big( (2b_N)^{k/2+1} + o(b_N^{k/2+1}) \big)
\\\nonumber &\overset{N\rightarrow\infty}\longrightarrow& 0 \ , 
\end{eqnarray}

\noindent since $b_N=o(N)$. Therefore, sum (\ref{Momente-O-N-Summe}) can be restricted to $\Gamma\in F_k$. Furthermore $\mu_k$ can be computed as

$$ \mu_k = \lim_{N\rightarrow\infty}\frac{1}{N\cdot (2b_N)^{k/2}} \cdot \sum_{\pi\in\mathcal{B}_{\frac{k}{2}}} \ \sum_{\Gamma\in F_k(\pi)} \mathbb{E}(X_N(\Gamma)) \ . $$

Since $\Gamma\in F_k$, we have $\mathbb{E}(X_N(\Gamma))=1$, because all random variables appear exactly twice. It is therefore sufficient to count the cardinality of $P|_F^{-1}(\pi)=F_k(\pi)$. In order to get an upper estimate, we embed $F_k$ into $\mathcal{F}_k$ and look for an upper estimate of $\# \mathcal{F}_k$. In order to do so, we apply inequality (\ref{ObereAbschaetzung-O-N}) as often as $k/2-1$ (which is possible according to Lemma \ref{BaumeHabenBlatter}). Using the notation $\tilde{\pi}=\{ \{ 1,2 \}\} \in\mathcal{B}_1$, we gain

$$ \# F_k \leq \# \mathcal{F}_k \leq (2b_N)^{k/2-1} \cdot \# P^{-1}(\tilde{\pi}) + o(b_N^{k/2+1}) \ .   $$

\noindent To estimate $\# P^{-1}(\tilde{\pi})$ we have to take into account, that as well adopted as non-adopted sequences can be used to be coloured. If $(G_1,G_2)$ is an adopted seqence, we have

$$ \# P^{-1}(\tilde{\pi}) \leq N \cdot 2b_N \ .  $$

\noindent If $(G_1,G_2)$ is not adopted, condition (\ref{BedingungO3}) gives \emph{only} $o(b_N^2)$ possibilities to colour these nodes. It follows

\begin{eqnarray}\nonumber
\lim_{N\rightarrow\infty}\frac{\# F_k}{N\cdot (2b_N)^{k/2}} &\leq& \lim_{N\rightarrow\infty}\frac{(2b_N)^{k/2-1} \cdot \big(N \cdot 2b_N +o(b_N^2) \big) + o(b_N^{k/2+1}) }{N\cdot (2b_N)^{k/2}} = 1 \ .
\end{eqnarray}

\noindent The Theorem is proved, if we can show, that

\begin{eqnarray}\label{FkIstMindestens1}
\lim_{N\rightarrow\infty}\frac{\# F_k}{N\cdot (2b_N)^{k/2}} \geq 1 \ .
\end{eqnarray}

\noindent This would imply

$$ \mu_k = \sum_{\pi\in\mathcal{B}_{\frac{k}{2}}} 1 \overset{\textnormal{Lemma \ref{CatalanZahl}}}= C_{\frac{k}{2}} \ , $$

\noindent while the odd moments vanish as explained above.

To prove (\ref{FkIstMindestens1}) we take $\pi\in\mathcal{B}_{\frac{k}{2}}$ and distinguish between two cases.

\begin{description}

	\item[Case (i):] $c(g_1)\in \{ k\cdot b_N+1, ..., N-k\cdot b_N \}$. Then we have $2b_N$ possibilities to colour $g_2$. Analogously to the proof of Lemma \ref{ObereAbschatzung} we colour the nodes $g_3,...,g_k$ successively. We further agree that for an index $l\in \{ 2,...,k \}$ there are two possibilities:

\begin{enumerate}

	\item 	$\exists \ m<l: \{ m,l \}\in\pi$. Then 
	
	$$ (c(g_m),c(g_{m+1}))\sim (c(g_l),c(g_{l+1})) $$
	
	\noindent is valid. Since $(g_1,...,g_k)$ is $\pi-$adodpted, we necessarily have $g_{l+1}=g_m$. Because among others $g_m$ has already been coloured, there is no choice to colour $g_{l+1}$.

	\item $\forall \  m<l: \{ m,l \}\not\in \pi$. Then
	
	$$ (c(g_m),c(g_{m+1})) \not\sim (c(g_l),c(g_{l+1})) $$

\noindent and because of condition (\ref{BedingungO2}) we have at least 

$$ 2b_N-(l-1)\cdot B \geq 2b_N -kB$$ 

\noindent choices to colour node $g_{l+1}$. 

\end{enumerate}

\noindent This shows that we have more than $2b_N -kB$ possibilities to colour each \emph{free} node. Since there are $\frac{k}{2}-1$ blocks left in the partition, we can freely colour $\frac{k}{2}-1$ nodes ($\{ 1, \beta_1 \}\in\pi$ was used to colour the first two nodes). The colour of the remaining 

$$k-2-\big(\frac{k}{2}-1\big)=\frac{k}{2}-1$$

\noindent nodes is constrained, because $(g_1,...,g_k)$ is $\pi-$adopted.

We gain

$$ \# F_k \geq (N-2k\cdot b_N) \cdot (2b_N) \cdot (2b_N -kB)^{k/2-1} \ . $$

	\item[Case (ii):] $c(g_1)\not\in \{ k\cdot b_N+1, ..., N-k\cdot b_N \}$. Then, the second node can be coloured in at least $b_N$ ways. The above argumentation can be used and we gain
	
	$$ \# F_k \geq (2k\cdot b_N) \cdot b_N \cdot (b_N -kB)^{k/2-1} \ , $$
	
	\noindent since every \emph{free} node can be coloured in at least $b_N$ ways.
\end{description}

\noindent Finally the above inequations give

\begin{eqnarray}\nonumber
\lim_{N\rightarrow\infty}\frac{\# F_k}{N\cdot (2b_N)^{k/2}} &\geq& \lim_{N\rightarrow\infty}\frac{1}{N\cdot (2b_N)^{k/2}} \cdot \big((N-2k\cdot b_N) \cdot (2b_N) \cdot (2b_N -kB)^{k/2-1} +
\\\nonumber && (2k\cdot b_N) \cdot b_N \cdot (b_N -kB)^{k/2-1} \big)
\\\nonumber &=& 1 \ ,
\end{eqnarray}

\noindent which completes the proof.

\end{proof}

\begin{remark}
As above we consider an equivalence relation with the minmal requirement 

$$(p,q)\sim (q,p) \textnormal{ for all } p,q\in \{ 1,...,N \} \ . $$

\noindent Assuming $X_{ij}\equiv X_{ji}$, this is again the Wigner case, with equivalence classes 

$$ [(p,q)]= \{ (p,q), (q,p) \} \ . $$

\noindent We will verify conditions (\ref{BedingungO1})-(\ref{BedingungO3}).

\begin{description}
	\item[Condition (\ref{BedingungO1}):] Let $p\in\{ 1,...,N\}$ be fixed. One has at most $2b_N$ possibilities to choose $q$ but then the pair $(r,s)$ is fixed from a set of two elements, which gives at most $4b_N$ choices.
	
	\item[Condition (\ref{BedingungO2}):] Let $(p,q,r)\in\{ 1,...,N\}^3$ be fixed. Then we have at most one possibility to choose $s$ if $r\in\{ p,q \}$. Otherwise we have none.
	
	\item[Condition (\ref{BedingungO3}):] We consider two pairs $(p,q)\sim(q,r)$. Then necessarily $p=r$ which is forbidden in the above condition. The set $\{ (p,q)\sim(q,r) \textnormal{ and } r\neq p \}$ is empty.
	
\end{description}

\end{remark}

\begin{remark}
The result for the Wigner case is discussed in \cite{BMP}.
\end{remark}

\newpage

\section{Weak Convergence in Probability}\label{section:WeakConvergence}

In this section we prove that the moments $\mu_k$ of the ensemble

$$ M^{(N)} := \frac{1}{\sqrt{N}} \left( \alpha\big(\frac{|i-j|}{N} \big) \cdot X^{(N)}_{ij} \right)_{(1\leq i,j \leq N)}  $$

\noindent do not only converge in \emph{expectation} but they also do \emph{weak in probability}. That means the following. Let $k\in\mathbb{N}$ be a natural number and set

$$ Y_N := \frac{1}{N} \cdot  \textnormal{tr }(M^{(N)})^k  \ . $$

\noindent Theorem \ref{TheoremMomentGleichSummeIntegrale} shows, that 

\begin{eqnarray}\nonumber
\mathbb{E}(Y_N) \overset{N\rightarrow\infty}\longrightarrow \left\lbrace \begin{array}{cc}
\sum_{\pi\in\mathcal{B}_{\frac{k}{2}}} J_\alpha(\pi) & \textnormal{for $k$ even}
\\ 0 & \textnormal{otherwise.}
\end{array}\right. 
\end{eqnarray}

\noindent what we refer to be a \emph{convergence in expectation}. Using Chebyshev's inequality we get

\begin{eqnarray}\nonumber
\mathbb{P}\left( |Y_N-\mathbb{E}(Y_N)|>\varepsilon \right) \leq \frac{\mathbb{V}(Y_N)}{\varepsilon^2} \ .
\end{eqnarray}

\noindent By proving that $\mathbb{V}(Y_N)\overset{N\rightarrow\infty}\longrightarrow 0$ we deduce, that

$$ \lim_{N\rightarrow\infty} \mathbb{P}\left( |Y_N-\mathbb{E}(Y_N)|>\varepsilon \right) = 0 \ \forall \ \varepsilon> 0 \ . $$

\noindent We refer this convergence to be \emph{weak in probability}.

\begin{theorem}\label{TheoremSchwacheKonvergenzinWkeitVollerFall}
Assume all the requirements from Theorem \ref{TheoremMomentGleichSummeIntegrale} for the symmetric ensemble

$$ M^{(N)} := \frac{1}{\sqrt{N}} \left( \alpha\big(\frac{|i-j|}{N} \big) \cdot X^{(N)}_{ij} \right)_{(1\leq i,j \leq N)} \ . $$

\noindent For every positive integer $k\in\mathbb{N}$ define

$$ Y_N^{(k)} := \frac{1}{N} \cdot  \textnormal{tr }(M^{(N)})^k \ . $$

\noindent Then we have

$$ \lim_{N\rightarrow\infty} \mathbb{V}(Y_N^{(k)}) = 0 \ \qquad \forall \ k\in\mathbb{N} \ .  $$ 
\end{theorem}

\begin{remark}
Before proving Theorem \ref{TheoremSchwacheKonvergenzinWkeitVollerFall} we remark the following asymptotic behaviour of $\# P^{-1}(\pi)$. For two sequences $(a_n)_{n\in\mathbb{N}}$ and $(b_n)_{n\in\mathbb{N}}$ we define

$$ a_n \sim b_n :\Longleftrightarrow \lim_{n\rightarrow\infty} \frac{a_n}{b_n} = \textnormal{const.} \ . $$

\noindent Consider again a partition $\pi:= \{ B_1, ..., B_{r} \}\in\mathcal{P}(\{ 1, ..., k \})$ and recall

\begin{eqnarray}\nonumber
\mathcal{E}^{(N)}_k(\pi) := \big\{ \Gamma:= \big((g_1,...,g_k),\pi,c \big) | \{ l,m \}\subset B_i\in\pi &\Leftrightarrow& \\\nonumber\big(c(g_m),c(g_{m+1})\big)&\sim& \big(c(g_l),c(g_{l+1})\big) \ , 
\\\nonumber  c: \{ g_1,...,g_k \} &\hookrightarrow& \{ 1,...,N \} , \ g_s\in\mathcal{G}  \big\} \ .
\end{eqnarray}

\noindent Lemma \ref{ObereAbschatzung} shows that

$$ \# P^{-1}(\pi) \sim N^{1+r} \ .  $$

\noindent Further the proof of Theorem \ref{TheoremMomentGleichSummeIntegrale} shows, that only non-crossing pair-partitions ($r=\frac{k}{2}$) contribute to sum (\ref{MomentSummeUeberE}). We will use this asymptotic behaviour for an upper estimate of the variance of $Y_N$.

\end{remark}

\begin{proof}(Theorem \ref{TheoremSchwacheKonvergenzinWkeitVollerFall})
Consider a subset $\pi\subset 2^{ \{ 1, ..., 2k \}}$. \emph{Define} 

$$ k+1:=1 \textnormal{ and } 2k+1:= k+1 \ . $$

\noindent As above we consider the set of paths, that is

\begin{eqnarray}\nonumber
\mathcal{E}^{(N)}_k(\pi) := \big\{ \Gamma:= \big((g_1,...,g_k,g_{k+1},...,g_{2k}),\pi,c \big) | \{ l,m \}\subset B_i\in\pi &\Leftrightarrow& \\\nonumber\big(c(g_m),c(g_{m+1})\big)&\sim& \big(c(g_l),c(g_{l+1})\big) \ , 
\\\nonumber  c: \{ g_1,...,g_{2k} \} &\hookrightarrow& \{ 1,...,N \} , \ g_s\in\mathcal{G}  \big\} \ .
\end{eqnarray}

\noindent For $\Gamma\in\mathcal{E}^{(N)}_k$ we define 

\begin{eqnarray}\nonumber
 X_N^{(1)}(\Gamma) &:=&\prod_{i=1}^k X_{c(g_i)c(g_{i+1})} \ , 
\\\nonumber X_N^{(2)}(\Gamma) &:=&\prod_{i=k+1}^{2k} X_{c(g_i)c(g_{i+1})} \ \textnormal{and}
\\\nonumber X_{N^2}(\Gamma) &:=&  X_N^{(1)}(\Gamma) \cdot  X_N^{(2)}(\Gamma) \ . 
\end{eqnarray}

\noindent We remark, that it is sufficient to consider $\pi$ as a partition since the definition of $\mathcal{E}_k^{(N)}$ gives an equivalence relation on $\{ 1, ..., 2k \}$.

Using these notations we can estimate the variance as

$$ \mathbb{V}(Y_N) \leq C_\alpha \cdot \frac{1}{N^{2+k}} \cdot \sum_{\pi\in\mathcal{P}(\{ 1, ..., 2k \})} \ \sum_{\Gamma\in\mathcal{E}_k^{(N)}} \mathbb{E}(X_{N^2}(\Gamma))-\mathbb{E}(X_N^{(1)})\cdot \mathbb{E}(X_N^{(2)}) \ , $$

\noindent with

$$ C_\alpha := \max_{\pi\in\mathcal{P}(\{1,...,2k\})} \{  |\alpha(\Gamma)| \ | \ \Gamma\in \mathcal{E}_k^{(N)} \} \ , $$ 

\noindent where $\alpha(\Gamma)$ is defined as in (\ref{MomentSummeUeberE}). This constant exists since $\alpha$ is bounded on $[0;1]$.

\noindent We next define the set of common edges for $X_N^{(1)}$ and $X_N^{(2)}$. For a given partition $\pi\in\mathcal{P}(\{ 1, ..., 2k \})$, an element $\Gamma\in\mathcal{E}_k^{(N)}$ and indices $1\leq i \leq k$ and $k+1\leq j \leq 2k$ we set


$$ \mathcal{K}_\Gamma := \{ \big( (g_i, g_{i+1}), (g_j, g_{j+1})  \big) \ | \ (c(g_i), c(g_{i+1}))\sim (c(g_j), c(g_{j+1})) \} \ .  $$

\begin{description}

	\item[Case (i):] $\mathcal{K}_\Gamma=\emptyset$. In this case, the random variables $X_N^{(1)}(\Gamma)$ and $X_N^{(2)}(\Gamma)$ are independent, which means that 
	
	$$ \mathbb{E}(X_{N^2}(\Gamma))= \mathbb{E}(X_N^{(1)}(\Gamma)) \cdot \mathbb{E}(X_N^{(2)}(\Gamma)) \ .   $$

\noindent Therefore the variance is equal to zero in this case.

	\item[Case (ii):] $\mathcal{K}_\Gamma\neq\emptyset$. As mentioned in the above remark there are at most $N^{1+k/2}$ possibilities to colour the first $k$ nodes, namely $g_1,...,g_k$. This is because there are no more than $k/2$ blocks which include sets of the form $\{ r,s \}$ with $r,s\leq k$ as the proof of Theorem \ref{TheoremMomentGleichSummeIntegrale} shows. For $i\leq k$ and $j\geq k+1$ we can choose a common edge 
	
	$$(c(g_i), c(g_{i+1})) \sim (c(g_j), c(g_{j+1})) \ .  $$

\noindent Because $g_i$ and $g_{i+1}$ have already been coloured we can freely colour $g_j$ which gives at most $N$ choices. Then the colour of $g_{j+1}$ is constrained by condition (\ref{AequivalenzrelationBedingung2}) and there are at most $B$ choices. 

\noindent The nodes $g_j$ and $g_{j+1}$ are now coloured. Using the above argument we can \emph{freely} colour $k/2-1$ of the remaining nodes in the set

$$ \{ g_{k+1}, ..., g_{2k} \} \backslash \{ g_j, g_{j+1} \} $$

\noindent in most $N$ ways. The reason is that there are again at most $k/2-1$ \emph{non-mixing} blocks left. These are blocks of the form $\{ r,s \}$ with $r,s\leq k$ or $r,s\geq k+1$. The other indices are again constrained by condition (\ref{AequivalenzrelationBedingung2}). This shows that

$$ \# P^{-1}(\pi) \sim N^{1+k/2} \cdot N \cdot B \cdot N^{k/2-1} \sim N^{k+1} \ . $$ 

\noindent It follows that there is a constant $c>0$ with

$$ \mathbb{V}(Y_N) \leq c \cdot \frac{1}{N^{2+k}} \cdot N^{k+1} \sim N^{-1} \overset{N\rightarrow\infty}\longrightarrow 0 \ . $$
\end{description}
\end{proof}

\begin{corollary}\label{KorollarSchwacheKonvInWkeitOvonN}
Assume all requirements from Theorem \ref{TheoremKleinOvonN} for the symmetric ensemble

$$ A^{(N)}:=\frac{1}{\sqrt{2b_N}} \left( \chi_{[0;b_N]} \big(|i-j| \big) \cdot X^{(N)}_{ij} \right)_{(1\leq i,j \leq N)} \ .  $$

\noindent For the above defined random variable

$$ Y_N^{(k)} := \frac{1}{N} \cdot  \textnormal{tr }(A^{(N)})^k \ . $$

\noindent we have

$$ \lim_{N\rightarrow\infty} \mathbb{V}(Y_N^{(k)}) = 0 \ \forall \ k \ .  $$ 

\end{corollary}

\begin{proof}
As above we define 

$$ k+1:=1 \textnormal{ and } 2k+1:= k+1 \ . $$

\noindent We then consider the set $\mathcal{F}_k$ of relevant cycles, that is

$$ \mathcal{F}^{(N)}_k(\pi) = \big\{ \Gamma=\big((g_1,...,g_k,...,g_{2k}),\pi,c\big)\in \mathcal{E}^{(N)}_k(\pi) | \ |c(g_m)-c(g_{m+1})| \leq b_N \ \forall \ 1\leq m \leq 2k  \big\} \ .  $$

\noindent Again, $\mathcal{E}^{(N)}_k$ is the set of cycles defined in Theorem \ref{TheoremSchwacheKonvergenzinWkeitVollerFall}. The proofs of Theorem \ref{TheoremKleinOvonN} and Theorem \ref{TheoremSchwacheKonvergenzinWkeitVollerFall} show, that we can compute and estimate the variance as

\begin{eqnarray}\nonumber 
\mathbb{V}(Y_N) &=& \frac{1}{N^{2}\cdot (2b_N)^k} \cdot \sum_{\pi\in\mathcal{P}(\{ 1, ..., 2k \})} \ \sum_{\Gamma\in\mathcal{F}_k^{(N)}} \mathbb{E}(X_{N^2}(\Gamma))-\mathbb{E}(X_N^{(1)})\cdot \mathbb{E}(X_N^{(2)})
\\\nonumber &\leq&  c_k \cdot \frac{1}{N^{2}\cdot (2b_N)^k} \cdot \#\mathcal{F}_k^{(N)} \ , 
\end{eqnarray}

\noindent with some constant $c_k$. This constant exists since all moments of the random variables are bounded. For a given partition $\pi\in\mathcal{P}(\{ 1, ..., 2k \})$, an element $\Gamma\in\mathcal{F}_k^{(N)}$ and indices $1\leq i \leq k$ and $k+1\leq j \leq 2k$ we again define the set of common edges for $X_N^{(1)}$ and $X_N^{(2)}$, that is

$$ \mathcal{K}_\Gamma := \{ \big( (g_i, g_{i+1}), (g_j, g_{j+1})  \big) \ | \ (c(g_i), c(g_{i+1}))\sim (c(g_j), c(g_{j+1})) \} \ .  $$

\begin{description}

	\item[Case (i):] $\mathcal{K}_\Gamma=\emptyset$. As above, the random variables $X_N^{(1)}(\Gamma)$ and $X_N^{(2)}(\Gamma)$ are independent, which means that the variance is equal to zero.

	\item[Case (ii):] $\mathcal{K}_\Gamma\neq\emptyset$. The proof of Theorem \ref{TheoremKleinOvonN} shows, that there are most $N\cdot (2b_N)^{k/2}$ possibilities to colour the first $k$ nodes, namely $g_1,...,g_k$. Again, this is because there are no more than $k/2$ \emph{non-mixing} blocks which include sets of the form $\{ r,s \}$ with $r,s\leq k$. For $i\leq k$ and $j\geq k+1$ we again choose a common edge 
	
	$$(c(g_i), c(g_{i+1})) \sim (c(g_j), c(g_{j+1})) \ .  $$
	
\noindent Because $g_i$ and $g_{i+1}$ have already been coloured we can freely colour $g_j$ which again gives most $N$ choices. Then the colour of $g_{j+1}$ is constrained by condition (\ref{AequivalenzrelationBedingung2}) and there are most $B$ choices. 

\noindent The nodes $g_j$ and $g_{j+1}$ are now coloured. The proofs of Theorems \ref{TheoremKleinOvonN} and \ref{TheoremSchwacheKonvergenzinWkeitVollerFall} show, that there are most $(2b_N)^{k/2-1}$ ways to \emph{freely} colour $k/2-1$ nodes in 

$$ \{ g_{k+1}, ..., g_{2k} \} \backslash \{ g_j, g_{j+1} \} \ . $$

\noindent This is because of most $k/2-1$ \emph{non-mixing} blocks that contain sets of the form $\{ r,s \}$ with $r,s> k$. The other colours are constrained by condition (\ref{BedingungO2}). This shows that

$$ \# P^{-1}(\pi) \sim N\cdot (2b_N)^{k/2} \cdot N \cdot B \cdot (2b_N)^{k/2-1} \sim N^2 \cdot (2b_N)^{k-1}  \ . $$ 

\noindent It follows that there is a constant $c>0$ with

$$ \mathbb{V}(Y_N) \leq c \cdot \frac{1}{N^{2}\cdot (2b_N)^{k}} \cdot N^{2}\cdot (2b_N)^{k-1} \sim \frac{1}{2b_N} \overset{N\rightarrow\infty}\longrightarrow 0 \ . $$
\end{description}

\end{proof}

\newpage

\addcontentsline{toc}{section}{Bibliography}
\bibliographystyle{alpha}

\begin{thebibliography}{BCM02i}


\bibitem[AGZ]{AGZ} G. W. Anderson, A. Guionnet, O. Zeitouni, An Introduction to Random Matrices. Cambridge University Press 2009.

\bibitem[ARN]{ARN} L. Arnold, On the asymptotic distribution of the eigenvalues of random matrices, J. Math. Anal. Appl., No. 20: 262-268, 1967.

\bibitem[BMP]{BMP} L. V. Bogachev, S. A. Molchanov, L. A. Pastur, On the Level Density of Random Band Matrices. Mathematical Notes 11/1991; 50(6):1232-1242.

\bibitem[DEI]{DEI} P. Deift, Orthogonal polynomials and random matrices: A Riemann-Hilbert approach. American Mathematical Society 1998.

\bibitem[GEO]{GEO} H.-O. Georgii, Stochastik. deGruyter 2004.

\bibitem[HiPe]{HiPe} F. Hiai und D. Petz. The semicircle law, free random variables and entropy. Mathematical Surveys and Monographs, 77. American Mathematical Society, Providence, RI, 2000.

\bibitem[HKW]{HKW} W. Hochst\"attler, W. Kirsch, S. Warzel, Semicircle Law for a matrix ensemble with dependent entries, J. Theor. Probab., published online 2015.


\bibitem[HCS]{HCS} K. Hofmann-Credner and Michael Stolz, Wigner theorems for random matrices with dependent entries: ensembles associated to symmetric spaces and sample covariance matrices, Elect. Comm. in Probab. 13(2008), 402-414.

\bibitem[HSBS]{HSBS} J. Schenker and H. Schulz-Baldes, Sermicircle Law and Freeness for Random Matrices with Symmetries or Correlations.

\bibitem[KIR]{KIR} W. Kirsch, An Invitation to random Schr\"odinger Operators, Panoramas et Syntheses 25, 1-119 (2008).


\bibitem[KIR2]{KIR2} W. Kirsch, A survey on the method of moments, Book in preparation.

\bibitem[KRI]{KRI} T. Kriecherbauer, Eine kurze und selektive Einf\"uhrung in die Theorie der zuf\"alligen Matrizen; in German, 2008, available at \emph{http://www.ruhr-uni-bochum.de/imperia/md/content/stochastik/og$g_-$th$o_-$v100811.pdf}.

\bibitem[MP]{MP} V. A. Marchenko, L. A. Pastur, Distribution of eigenvalues in certain sets of random matrices, Mat. Sb., 72 (114), No. 4, 507-536 (1967).

\bibitem[NKP]{NKP} R. Speicher, Lecture on free probability theory; in German, 1997/98, available at \emph{http://www.mast.queensu.ca//$\sim$speicher/}.

\bibitem[PAS1]{PAS1} L. A. Pastur, The spectra of random selfadjoint operators, Usp. Mat. Nauk, 28, No. 1, 3-64 (1973).

\bibitem[PS]{PS} L. A. Pastur, M Shcherbina, Bulk universality and related properties of Hermitian matrix models. J. Sat. Phys., No 130: 205-250, 2008.

\bibitem[TAO]{TAO} T. Tao, Topics in random matrix theory, American Mathematical Society 2012.

\bibitem[WER]{WER} D. Werner, Funktionalanalysis. Springer 2005.

\bibitem[WIG1]{WIG} E. P. Wigner, On the distribution of the roots of certain symmetric matrices. Ann. Math., 67, No. 2, 325-327, 1958.

\bibitem[WIG2]{WIG2} E. P. Wigner, Characteristic vectors of bordered matrices with infinite dimensions, Ann. of Math. 62, 548-564 (1955).

\end{thebibliography}

\newpage\thispagestyle{empty}

{\Huge Eidesstattliche Erkl\"arung}

\vskip 3cm

\noindent Ich versichere, dass ich die vorliegende Dissertation selbst\"andig verfasst, andere als die angegebenen Quellen und Hilfsmittel nicht benutzt und mich auch sonst keiner unerlaubtem Hilfe bedient habe.

\vskip 6cm

\noindent Hagen, Februar 2016

\noindent Riccardo Catalano

\end{document}